\documentclass[11pt,reqno]{article}

\usepackage{mathtools,amssymb,amsthm,xcolor}
\usepackage[shortlabels]{enumitem}
\setlist[enumerate]{font=\upshape}
\usepackage{bookmark}

\usepackage[margin=1.2in]{geometry}
\newcommand*{\bX}{\mathbf{X}}

\newcommand*{\bbT}{\mathbb{T}}
\newcommand*{\EE}{\mathbb{E}}

\newcommand*{\cE}{\mathcal{E}}

\newcommand*{\cG}{\mathcal{G}}

\newcommand*{\CC}{\mathbb{C}}

\newcommand*{\R}{\mathbb{R}}

\renewcommand*{\Pr}{\mathbb{P}}
\newcommand*{\eps}{\varepsilon}
\newcommand*{\del}{\delta}
\newcommand*{\lam}{\lambda}

\DeclareMathOperator{\var}{var}

\DeclareMathOperator{\cut}{cut}
\DeclareMathOperator{\artanh}{artanh}
\DeclareMathOperator{\sech}{sech}

\newcommand*{\Zfixed}{Z^{\mathrm{fix}}}

\newcommand*{\csproblem}[1]{\textup{\textsc{#1}}}

\newtheorem*{theorem*}{Theorem}
\newtheorem{theorem}{Theorem}
\newtheorem{lemma}[theorem]{Lemma}

\newtheorem*{defn*}{Definition}
\newtheorem{prop}[theorem]{Proposition}
\newtheorem*{prop*}{Proposition}
\newtheorem{conj}{Conjecture}
\newtheorem*{conj*}{Conjecture}

\newtheorem{question}{Question}
\newtheorem*{fact*}{Fact}

\begin{document}
\title{Computational thresholds for the fixed-magnetization\\ Ising model}

\author{%
Charlie Carlson\thanks{Department of Computer Science, University of Colorado Boulder,  chca0914@colorado.edu.}
\and
Ewan Davies\thanks{Department of Computer Science, University of Colorado Boulder,  ewan.davies@colorado.edu.}
\and
Alexandra Kolla\thanks{Department of Computer Science, University of Colorado Boulder, Department of Computer Science and Engineering, University of California Santa Cruz, alexandra.kolla@colorado.edu. Supported in part by NSF grant CCF-1452923}
\and
Will Perkins\thanks{Department of Mathematics, Statistics, and Computer Science, University of Illinois at Chicago, math@willperkins.org.  Supported in part by NSF grants DMS-1847451 and CCF-1934915.}}

\date{\today}

\maketitle
\thispagestyle{empty}

\begin{abstract}
The ferromagnetic Ising model is a model of a magnetic material and a central topic in statistical physics.  It also plays a starring role in the algorithmic study of approximate counting: approximating the partition function of the ferromagnetic Ising model with uniform external field is tractable at all temperatures and on all graphs, due to the randomized algorithm of Jerrum and Sinclair.   

Here we show that hidden inside the model are hard computational problems.  For the class of bounded-degree graphs we find computational thresholds for the approximate counting and sampling problems for the ferromagnetic Ising model at fixed magnetization (that is, fixing the number of $+1$ and $-1$ spins). 

In particular, letting $\beta_c(\Delta)$ denote the critical inverse temperature of the zero-field Ising model on the infinite $\Delta$-regular tree, and  $\eta_{\Delta,\beta,1}^+$ denote the mean magnetization of the zero-field $+$ measure on the infinite $\Delta$-regular tree at inverse temperature $\beta$, we prove, for the class of graphs of maximum degree $\Delta$:
\begin{enumerate}
\item For $\beta < \beta_c(\Delta)$ there is an FPRAS and efficient sampling scheme for the fixed-magnetization Ising model for all magnetizations $\eta$.
\item For $\beta >  \beta_c(\Delta)$, there is an FPRAS and efficient sampling scheme for the fixed-magnetization Ising model for magnetizations $\eta$ such that  $|\eta| >\eta_{\Delta,\beta,1}^+ $.
\item For $\beta >  \beta_c(\Delta)$, there is no FPRAS for the fixed-magnetization Ising model for magnetizations $\eta$ such that  $|\eta| <\eta_{\Delta,\beta,1}^+ $ unless NP=RP\@.
\end{enumerate}
\end{abstract}

\maketitle

\clearpage
\pagenumbering{arabic} 

\section{Introduction}

The Ising model is a mathematical model of a magnetic material, fundamental in the study of phase transitions in statistical physics.  
The Ising model is a probability distribution over cuts in a graph, and its partition function  is the weighted sum over all cuts in the graph, connecting the physics of the model to combinatorial structures in computer science.
In the field of approximate counting in computer science, the ferromagnetic Ising model plays a special role along with the monomer-dimer model as models for which approximating the partition function is tractable on all graphs and at all temperatures~\cite{jerrum1989approximating,JS93}.  

Conditioning on the magnetization of the model corresponds to fixing the balance of the random cut  generated. In particular, at zero magnetization (an equal number of plus and minus spins), the Ising model is a probability distribution on bisections of the graph. 
In the study of spin models on sparse random graphs in physics, it has long been known that conditioning on zero magnetization can turn a ferromagnetic system into a glassy system~\cite{mezard1987mean} (i.e.\ fixing the magnetization can drastically change the model and induce slow dynamics).  This suggests that lurking inside the tractable computational problems associated to the Ising model there may be hard problems accessible by fixing the magnetization. 

We make this idea concrete in a complexity-theoretic sense by reducing NP-hard balanced cut problems to approximating the partition function of the Ising model at fixed magnetization. 
Specifically we find computational thresholds for approximate counting and sampling in the ferromagnetic Ising model at fixed magnetization on bounded degree graphs.  When the inverse temperature $\beta$ is small (smaller that the critical $\beta$ on the infinite $\Delta$-regular tree) there are efficient algorithms at all magnetizations.  When $\beta$ is large (larger than the critical $\beta$) then there is a computational threshold: for magnetizations $\eta$ small in absolute value the computational problems are hard; for $\eta$ large in absolute value the problems are tractable.  

We first define the Ising model and the relevant properties of the model on the infinite regular tree, then state our main results.

\subsection{The Ising model on graphs and trees}

The Ising model on a finite graph $G$ at inverse temperature $\beta $ and activity $\lam$ is the probability distribution $\mu_{G,\beta,\lambda}$ on $ \Sigma_G := \{ \pm 1\}^{V(G)}$ defined by
\begin{align*}
\mu_{G,\beta,\lambda}(\sigma) = \frac{ e^{ \frac{\beta}{2} \sum_{(u,v) \in E(G)} \sigma_u \sigma _v} \lam^{ M(\sigma)}  }{  Z_G(\beta,\lam) } 
\end{align*}
where  $M( \sigma ) = \sum_{v \in V(G)} \sigma_v$ and
\begin{align*}
Z_G(\beta,\lam)  =  \sum_{\sigma \in \Sigma_G}  e^{\frac{\beta}{2} \sum_{(u,v) \in E(G)} \sigma_u \sigma _v} \lam^{ M(\sigma)}  \,.
\end{align*}
The probability distribution $\mu_{G,\beta, \lam}$ is the Gibbs measure and $Z_G(\beta,\lam)$ is the partition function.  When $\beta \ge 0$ the model is ferromagnetic, and we will always assume this in what follows.  In statistical physics the activity is often written as $\lam = e^{h}$ where $h$ is the external field, and so we will call the unbiased case $\lam =1$ the zero-field Ising model.   

The quantity $M(\sigma)$ is the \emph{magnetization} of the configuration $\sigma$.  The normalized \emph{mean magnetization} of the Ising model is 
\[ \eta_{G}(\beta,\lam)  = \frac{ \langle M(\sigma)\rangle_{G,\beta,\lam} }{|V(G)|} \,,\]
where $\langle \cdot \rangle_{G,\beta,\lam}$ denotes expectation with respect to the Ising model.

We can also define the Ising model at fixed magnetization.  For $ k \equiv |V(G)| \mod 2$, $|k| \le |V(G)|$, let $\Sigma_G(k) = \{ \sigma \in \Sigma_G : M(\sigma) =k \}$ be the subset of Ising configurations with magnetization $k$.  Then the Ising model on $G$ at inverse temperature $\beta$ and fixed magnetization $k$ is the distribution $\nu_{G,\beta,k}$ on $\Sigma_G(k)$ defined by
\begin{align*}
\nu_{G,\beta,k} (\sigma) = \frac{  e^{\frac{\beta}{2} \sum_{(u,v) \in E(G)} \sigma_u \sigma _v}    }{ Z_G^{\mathrm{fix}}(\beta ,k)    }
\end{align*}
where
\begin{align*}
Z_G^{\mathrm{fix}}(\beta ,k)  = \sum _{\sigma \in \Sigma_G(k) }  e^{\frac{\beta}{2} \sum_{(u,v) \in E(G)} \sigma_u \sigma _v}  \,.
\end{align*}
The distribution  $\nu_{G,\beta,k}$ is simply the Ising model at inverse temperature $\beta$ (and arbitrary activity $\lam >0$) conditioned on the event  $ \sigma \in \Sigma_G(k)$.  The fixed-magnetization partition function $Z_G^{\mathrm{fix}}(\beta ,k) $ is the coefficient of $\lam^{k}$ when interpreting $ Z_G(\beta,\lam)$ as a Laurent polynomial in $\lam$.

The Ising model can be defined on the infinite $\Delta$-regular tree $\bbT_{\Delta}$ via the DLR equations~\cite{Dob68,LR69} or as a weak limit of Ising models on finite-depth trees with given boundary conditions.
Infinite regular trees are important in computer science as `optimal' expanders, and here we use a known relationship between the Ising model on random regular graphs and on the infinite tree.
Depending on the parameters $\beta, \lam$ there may be a unique infinite-volume Gibbs measure on $\bbT_{\Delta}$ or there may be multiple measures. The critical inverse temperature is $\beta_c(\Delta) = \log  \frac{\Delta}{\Delta-2} $: for $\beta< \beta_c$ there is a unique Gibbs measure for all $\lam$ and for $\beta>\beta_c$ there can be multiple measures if $\lam$ is close enough to $1$~\cite{lyons1989ising}.    We will be interested in one particular  Gibbs measure on $\bbT_{\Delta}$, the `$+$' measure induced by the weak limit of finite-depth trees with the all $+$ boundary conditions.  We denote this measure $\mu_{\Delta, \beta, \lam}^+$.  By the FKG inequality $\mu_{\Delta, \beta, \lam}^+$ stochastically dominates all other Gibbs measures on $\bbT_\Delta$ with the same parameters.

We let $\eta_{\Delta, \beta,\lam}^+$ denote the expected value of the spin at the root of $\bbT_{\Delta}$ under $\mu_{\Delta, \beta, \lam}^+$ (equivalently, the expected value of the spin at any fixed vertex since $\mu_{\Delta, \beta, \lam}^+$ is translation invariant).  
Then the magnetization of the measure $\mu_{\Delta,\beta,\lam}^+$ is 
\[ \eta_{\Delta, \beta,\lam}^+ = \tanh\big(L^*+\artanh(\tanh L^*\tanh\tfrac{\beta}{2})\big), \]
where $L^*$ is the largest solution to 
\[ L^* = \log\lam + (\Delta-1)\artanh(\tanh L^*\tanh\tfrac{\beta}{2}). \]
See Section~\ref{secExtremal} for more details and a derivation.

The phase transition on $\bbT_{\Delta}$ manifests itself via the following `spontaneous magnetization' phenomenon~\cite{lyons1989ising}:
\begin{enumerate}
\item For $\beta < \beta_c(\Delta)$, $ \eta_{\Delta, \beta,1}^+ =0$.
\item For $\beta > \beta_c(\Delta)$, $ \eta_{\Delta, \beta,1}^+ >0$.
\end{enumerate}

\subsection{Computational problems and computational thresholds}

There are two main computational problems associated to a spin model like the Ising model.  The approximate counting problem asks for an $\eps$-relative approximation to the partition function $Z_G$; that is,  a number $\hat Z$ so that $(1- \eps) Z_G \le \hat Z \le (1+\eps ) Z_G$.  An FPTAS is an algorithm that provides such an approximation and runs in time polynomial in $|V(G)|$ and $1/\eps$.  An FPRAS is a randomized algorithm that provides such an approximation with probability at least $2/3$ and runs in time polynomial in $|V(G)|$ and $1/\eps$.  The approximate sampling problem is to output a sample $\sigma$ with distribution $\hat \mu$ so that $\| \mu_G - \hat \mu \|_{TV} <\eps$.  An efficient sampling scheme is a randomized algorithm that satisfies this guarantee and runs in time polynomial in $|V(G)|$ and $\log (1/\eps)$\footnote{Sometimes the required dependence of the running time for an efficient approximate sampler is taken to be polynomial in $1/\eps$ instead of $\log(1/\eps)$; we use the stronger definition here.}.

Jerrum and Sinclair gave an FPRAS for the ferromagnetic Ising model for all graphs, all inverse temperatures $\beta$, and all choices of the activity $\lam$~\cite{JS93}\footnote{In fact the algorithm works in the case of non-uniform activities, as long as they are consistent: all at least $1$ or all at most $1$.  The general case of approximating the partition function with non-uniform activities is \#BIS-hard~\cite{goldberg2007complexity}.}. Via self-reducibility of the random cluster representation of the Ising model, this gives an efficient sampling scheme as well~\cite{randall1999sampling}. 

On the other hand, for the anti-ferromagnetic Ising model (and the hard-core model of weighted independent sets), the approximate counting and sampling problems are \textup{NP}-hard in general, and for the class of bounded degree graphs precise computational thresholds are known.   The results of Weitz~\cite{Wei06}, Sly~\cite{Sly10},  Sly--Sun~\cite{SS14}, Galanis--{\v S}tefankovi{\v c}--Vigoda~\cite{galanis2016inapproximability}, and Sinclair--Srivastava--Thurley~\cite{sinclair2014approximation} show that for these models (and for $\beta$ large enough in the case of the anti-ferromagnetic Ising model) there is a computational threshold at some critical activity $\lam_c = \lam_c(\Delta,\beta)$. In the case of the hard-core model, there is an FPTAS for $Z_G(\lam)$ for  $\lam < \lam_c$ and  graphs $G$ of maximum degree $\Delta$ while for $\lam>\lam_c(\Delta)$ there is no FPRAS unless \textup{NP=RP}.

For the ferromagnetic Ising model there are no such computational thresholds.   But one can ask instead for approximation algorithms for coefficients of the partition function or approximate sampling algorithms for the Ising model at fixed magnetization. 

For $\beta \ge 0$ and $\eta \in [-1,1]$, let \csproblem{Fixed-Ising}$(G,\beta,\eta)$ be the problem of computing the partition function $\Zfixed_G(\beta, k)$ of the $n$-vertex graph $G$, where $k$ is the largest integer such that $k\equiv n \mod 2$ and $k\le \eta n$.   In other words, $k = 2 \lfloor (\eta+1)n/2 \rfloor -n$. 
The associated sampling problem is to sample spin assignments from the measure $\nu_{G,\beta,k}$. 
The restriction on the parity of $k$ is simply to ensure that configurations of magnetization $k$ exist.   Abusing notation slightly we will refer to both $\eta$ and $k$ as the magnetization, but it will be clear from context what is meant.

This is the setting of the Kawasaki dynamics for the Ising model~\cite{kawasaki1966diffusion,kawasaki1972kinetics}: a \emph{conservative} dynamics with stationary distribution $\nu_{G, \beta, k}$ that at each step proposes a swap of nearest-neighbor spins.    Understanding the convergence properties of the Kawasaki dynamics on subsets of $\mathbb Z^d$ is a deep mathematical problem~\cite{lu1993spectral,yau1996logarithmic,cancrini1999spectral,cancrini2000spectral}.  In this paper we address the problem on general graphs from the perspective of computational complexity.

\subsection{Our results}

In what follows we always assume $\beta \ge0$ and $\Delta \ge 3$. 
When $\beta < \beta_c (\Delta)$ we give efficient approximate counting and sampling  algorithms for all magnetizations.

\begin{theorem}\label{thmAlgFixedSubcriticalBeta}
    Let $\Delta\ge 3$ and $\beta < \beta_c(\Delta)$. Then for all $\eta\in[-1,1]$ there is an FPRAS and efficient sampling scheme for \csproblem{Fixed-Ising}$(G,\beta,\eta)$ for graphs of maximum degree $\Delta$. 
\end{theorem}

Theorem~\ref{thmAlgFixedSubcriticalBeta} can be deduced fairly easily from known results, essentially following the framework of~\cite{davies2021approximatelyICALP}.   To sample from configurations with a given magnetization, we follow the standard approach of finding a suitable activity parameter for the Gibbs measure $\mu_{G,\beta,\lam}$ so that the probability of hitting the desired magnetization is not too small (at least inverse polynomial), and then sampling from the Ising model, rejecting samples until we obtain one with the correct magnetization.  Because  efficient sampling algorithms for the Ising model exist for all $\beta, \lam$ this approach works provided that a suitable activity parameter exists.  By continuity, there is an activity that gives the correct mean magnetization, and because the partition function (as a function of $\lam$) is uniformly zero-free in a sector in the complex plane~\cite{PR20}, the magnetization obeys a central limit theorem~\cite{MS19}, giving the required inverse polynomial lower bound.  

The main results of the paper are for the supercritical case, $\beta>\beta_c(\Delta)$. Here we prove that there is a computational threshold at an explicit $\eta_c = \eta_c(\Delta,\beta) \in (0,1)$ so that approximation is hard for $|\eta| < \eta_c$ but tractable for $|\eta| > \eta_c$.  In fact, $\eta_c(\Delta,\beta) = \eta_{\Delta,\beta,1}^+$, the mean magnetization of the zero-field $+$ measure on $\bbT_\Delta$.

\begin{theorem}\label{thmFixedSupercriticalBeta}
    Let $\Delta\ge 3$, $\beta>\beta_c(\Delta)$, and $\eta_c = \eta_{\Delta,\beta,1}^+$.
\begin{enumerate}[label={\textup{(\alph*)}}]
    \item\label{algFixedSuper} For all $\eta$ with $|\eta|>\eta_c$ there is an FPRAS and efficient sampling scheme for \csproblem{Fixed-Ising}$(G,\beta,\eta)$ for graphs of maximum degree $\Delta$. 
    \item\label{hardFixedSuper} Unless \textup{NP=RP}, for all $\eta$ with $|\eta| < \eta_c$ there is no FPRAS for \csproblem{Fixed-Ising}$(G,\beta,\eta)$ for graphs of maximum degree $\Delta$.
\end{enumerate}
\end{theorem}

In~\ref{hardFixedSuper} our proof in fact shows that given $\Delta$, $\beta$ there is some $\zeta>0$ such that unless NP=RP, there is no polynomial-time algorithm for \csproblem{Fixed-Ising}$(G,\beta,\eta)$ which achieves a multiplicative approximation of $e^{n^{\zeta}}$ on $n$-vertex graphs $G$ of maximum degree $\Delta$.

The infinite regular tree plays several  roles in the proof of Theorem~\ref{thmFixedSupercriticalBeta}.   
For the hardness results, non-uniqueness for the zero-field Ising model on the tree at $\beta > \beta_c$ corresponds to `phase coexistence' of the model on the random $\Delta$-regular graph~\cite{dembo2010ising}.  Phase coexistence allows us to use random graphs as gadgets, as Sly does in establishing a computational threshold for the hard-core model~\cite{Sly10} (and as is done in subsequent hardness proofs, e.g.~\cite{SS14,GSV15,cai2016hardness}).  Our analysis of the hardness reduction requires new techniques to account for the fixed-magnetization constraint; we give an overview of the approach in the next section. 

For the algorithmic results, the $+$ measure on the infinite regular tree is the solution to a problem from  extremal graph theory that is essential for the proof of Theorem~\ref{thmFixedSupercriticalBeta}.   

For the ferromagnetic Ising model with activity $\lam>1$, what is the maximum mean magnetization over all graphs of maximum degree $\Delta$?  We prove that the magnetization of the $+$ measure on the infinite $\Delta$-regular tree is an upper bound, and this value is approached by that of the random $\Delta$-regular graph in the $n \to \infty$ limit.  The following result is the main combinatorial result of our paper.

\begin{theorem}\label{thmExtremal}
    For all graphs $G$ of maximum degree $\Delta$, all $\lam \ge1$, and all $\beta \ge0$,
    \[ \eta_{G}(\beta,\lam)\le \eta_{\Delta,\beta,\lam}^+ \, .\]
\end{theorem}

By integrating the mean magnetization from $\lam=1$ to $\infty$, this theorem implies the $\Delta$-regular case of a result of Ruozzi which states that the `Bethe approximation'  is a lower bound on the normalized partition function of the ferromagnetic Ising model~\cite{ruozzi2012bethe}.   In combinatorics, results of this type belong to the field of extremal problems for bounded-degree graphs:   maximizing or minimizing observables of statistical physics models over given classes of graphs, like the occupancy fraction of the hard-core or monomer-dimer models~\cite{DJPR17}.  The area is surveyed by Zhao in~\cite{zhao2017extremal} and Csikv{\'a}ri describes several cases in which the optimal bound on a partition function is given by an analogous quantity on an infinite regular tree~\cite{csikvari2016extremal}.  Bounds on observables such as the mean magnetization or occupancy fraction are stronger than bounds on the partition function, and to the best of our knowledge Theorem~\ref{thmExtremal} is the first case in which the infinite tree is proved to be extremal for an observable. 

Theorem~\ref{thmExtremal} implies the following extremal spontaneous magnetization result, which is what we use to guarantee the effectiveness of our algorithm.
Define
\[ \eta^*(\Delta,\beta) = \lim_{\lam\to1^+}\sup_{G\in \cG_{\Delta}}\eta_G(\beta,\lam),\]
where  $\cG_\Delta$ is the class of graphs of maximum degree $\Delta$.
Then $\eta^*(\Delta,\beta) =  \eta_{\Delta,\beta,1}^+$.
The lower bound comes from taking a sequence of random $\Delta$-regular graphs, while the upper bound follows from Theorem~\ref{thmExtremal}.    We describe in the next section the content of our algorithmic results for $\beta>\beta_c$: that $\eta_c(\Delta,\beta)=\eta^*(\Delta,\beta)  = \eta_{\Delta,\beta,1}^+$.

\subsection{Overview of the techniques}

\subsubsection{Algorithms}

For the algorithmic results of Theorem~\ref{thmFixedSupercriticalBeta}, we aim to apply the same type of algorithm as in Theorem~\ref{thmAlgFixedSubcriticalBeta}: find an activity $\lam$ so the mean magnetization is close to the target magnetization, and prove that the probability of hitting the mean is not too small.  Again by continuity, there is an activity with the correct mean magnetization, but the distribution may not be concentrated around its mean. For instance, taking $\lam=1$  gives $0$ mean magnetization by symmetry, but if $\beta>\beta_c$, then hitting $0$ magnetization on the random regular graph is exponentially unlikely.  So our question becomes: given an arbitrary graph of maximum degree $\Delta$ and a desired magnetization $\eta$, is the magnetization under $\mu_{G,\beta,\lam}$ guaranteed to be concentrated around its mean when $\lam$ is chosen so that the mean magnetization is (close to) $\eta$?   The answer to this question   is given by the Lee--Yang theorem~\cite{lee1952statistical} in combination with  Theorem~\ref{thmExtremal}, which guarantees that to achieve a mean magnetization $\eta > \eta^*(\Delta,\beta)$ we can pick an activity $\lam$ bounded away from $1$ independent of $n$.  The Lee--Yang theorem then gives the zero-freeness result that provides us with the required central limit theorem.

Our proof of Theorem~\ref{thmExtremal} is an extension of an approach used by Krinsky~\cite{krinsky1975bethe} to prove the result for infinite lattices like $\mathbb Z^d$ (or more generally graphs satisfying vertex and edge transitivity).  
The theorem (and the paper~\cite{krinsky1975bethe} that inspired it) may be of independent interest in combinatorics and algorithms.  The proof of Theorem~\ref{thmExtremal} relies heavily on correlation inequalities, namely the GKS inequalities~\cite{griffiths1967correlations,kelly1968general}, and identities due to Thompson~\cite{thompson1971upper}.  The techniques are distinct from previous approaches in this area of extremal graph theory such as the entropy method~\cite{kahn2001entropy}, occupancy method~\cite{DJPR17}, and inductive approaches~\cite{csikvari2017lower,sah2020reverse}.

\subsubsection{Hardness}
To prove a matching hardness result, we must overcome the barrier of the tractability of approximating the  Ising partition function. 
This rules out the approach used in~\cite{davies2021approximatelyICALP} for proving hardness of approximating the number of independent sets of a given size, namely reducing approximating the  partition function to approximating a fixed coefficient of the  partition function. 
Instead, we use the fact that imposing the fixed-magnetization constraint fundamentally alters the behavior of the model.  When highly connected components of a graph are connected with a relatively sparse set of edges, the fixed-magnetization, zero-field ferromagnetic Ising model exhibits a kind of global anti-ferromagnetic behavior due to the constraint on the magnetization: the spins on each highly connected component will align, but the number of components that pick each spin will be essentially determined by the constraint.  This behavior is what allows us to prove hardness.    We use a probabilistic analysis of the fixed-magnetization Ising model to show that a gadget construction based on that of~\cite{Sly10} can be used to reduce an NP-hard cut problem to approximating the fixed-magnetization Ising partition function.  To illustrate our methods we sketch a simplified version of the proof for zero magnetization.

Similar to previous approaches, our gadget $G$ is essentially a random $\Delta$-regular bipartite graph with some edges removed and trees attached to create `terminal vertices' of degree $\Delta-1$.  Given an instance $H$ of \csproblem{Min-Bisection}, we replace each vertex of $H$ by a copy of the gadget $G$ and then join a number of terminal vertices of the appropriate copies of the gadget graph for each edge of $H$. When $\beta>\beta_c$, the Ising model on a single gadget $G$ exhibits \emph{phase coexistence}, with a bimodal distribution of either many more $+$ spins than $-$ spins or vice versa.  The phase coexistence property of each gadget is so strong that when we take the collection of gadgets joined by the crossing edges and condition the Ising model on zero magnetization, the phase coexistence property on each gadget persists, and zero-magnetization is achieved (with high probability) by having an equal number of gadgets in each phase.  
Showing this involves proving a local central limit theorem and large deviation results for the magnetization of a collection of gadgets conditioned on an arbitrary spin assignment to the  set of terminal vertices.  This shows that the dominant contribution to the zero-magnetization partition function is given by configurations whose gadget phase assignments encode minimum bisections of $H$, and this in turn implies that a good approximation algorithm for the partition function can recover a minimum bisection. 

The proof of the local central limit theorem conditioned on the phases of the gadgets is a new technical ingredient in our proof.  It involves bounding the moments of the magnetization on a single gadget, conditioned on a phase, and employing a Fourier analytic proof of a local central limit theorem.

The full proof and the general case of $\eta \ne 0$ are only slightly more complex. 
Broadly, the same approach works except we reduce from a generalization of \csproblem{Min-Bisection}, $\gamma$-\csproblem{Min-Exact-Balanced-Cut} ($\gamma$-\csproblem{MEBC}), that requires the partition of a vertex set of size $N$ to have part sizes $\lfloor \gamma N \rfloor$ and $\lceil (1-\gamma)N \rceil$.  It is convenient to add to the collection of gadget graphs some isolated vertices which smooth out certain parts of the analysis.   In particular, it helps in proving the local central limit theorem. 
We choose $\gamma$ as a function of $\Delta$, $\beta$, and $\eta$, and we prove that when the Ising model on the collection of gadget graphs is conditioned to have  magnetization $\eta$, with high probability the phases of the gadgets are split in fractions $\gamma$ and $1-\gamma$. 
Then  a good approximation algorithm for the $\eta$-magnetization partition function can recover a minimum $\gamma$-balanced cut.

\subsection{Related work}

The algorithmic problem of sampling configurations of a fixed magnetization (or fixed size, in the case of independent sets) is the problem of sampling from the `canonical ensemble' in the language of statistical physics (in contrast to the `grand canonical ensemble' of the usual Ising or hard-core model).  Work on this problem goes back to the very first Markov Chain Monte Carlo algorithm designed to sample from the canonical ensemble of hard spheres~\cite{metropolis1953equation}.    Conservative dynamics such as these are still among the most used in current scientific applications (e.g.,~\cite{bernard2009event}).  Grand canonical ensembles are generally more amenable to mathematical analysis due to their conditional independence properties, and much is known about both specific algorithms for sampling from these distributions (e.g. Glauber dynamics~\cite{mossel2013exact}, random-cluster dynamics~\cite{guo2018random})  and about the computational complexity of the approximate counting and sampling problems for these models.  

The computational complexity of approximately counting and sampling independent sets of a given size in bounded-degree graphs was recently addressed by Davies and Perkins who proved a computation threshold for these problems~\cite{davies2021approximatelyICALP}.  As in Theorem~\ref{thmFixedSupercriticalBeta}, the threshold is given in terms of an extremal graph theory problem: that of minimizing the occupancy fraction over $G \in \cG_{\Delta}$.  Faster algorithms and an FPTAS up to the threshold for this problem were recently given in~\cite{jain2021approximate}.

The use of random graphs as gadgets in hardness reductions was pioneered by Dyer, Frieze, and Jerrum~\cite{dyer2002counting} and used by Sly in identifying the computational threshold for the hard-core model~\cite{Sly10}, with further applications in~\cite{SS14,GSV15,cai2016hardness,GSVY16} among others.  In particular, a detailed understanding of the moments of the partition function $Z_G$ for random regular graphs is now known, and, via the small subgraph conditioning method, concentration results for $Z_G$.  We use this understanding extensively in Section~\ref{secHardness}.

Finally, the Ising model at fixed magnetization has been studied extensively in both mathematics and  physics,  on $\mathbb Z^d$ and on random graphs~\cite{mezard1987mean}.  Conditioning the ferromagnetic Ising model on zero magnetization has the effect of introducing `frustration': the impossibility of satisfying all edge constraints simultaneously.   

At zero temperature ($\beta =\infty$), the zero-magnetization Ising model is simply the uniform distribution on min-bisections of a graph; finding the size of the min bisection has long been known to be NP-hard~\cite{garey1974some}. The min-bisection problem is also studied on random graphs from the perspective of statistical physics~\cite{percus2008peculiar,zdeborova2010conjecture,diaz2004computation,dembo2017extremal}.  Our work is an exploration of the worst-case computational complexity of the positive temperature regime of this problem.

\subsection{Questions and future directions}

Though we do  not pursue it in this extended abstract, it is likely that the techniques of Jain, Perkins, Sah, and Sawhney~\cite{jain2021approximate} can be used to improve the algorithmic results of Theorems~\ref{thmAlgFixedSubcriticalBeta} and~\ref{thmFixedSupercriticalBeta} in two ways:
\begin{enumerate}
\item Obtain an FPTAS (efficient deterministic approximation algorithm) for \csproblem{Fixed-Ising}$(G,\beta,\eta)$  for the same range of parameters for which we obtain an FPRAS.
\item Improve the running time of our approximate sampling algorithm to $\tilde O( n \log n)$. 
\end{enumerate}

We have shown here  a computational threshold for the fixed-magnetization Ising model.  One can also ask what is achievable with a specific algorithm widely used in scientific applications, namely the Kawasaki dynamics.  We conjecture that the Kawasaki dynamics mix rapidly on all graphs of maximum degree $\Delta$ for the same set of parameters for which we provide an FPRAS.  In fact there are two versions of the Kawasaki dynamics: the \emph{local flip} dynamics in which at each step a swap of spins across an edge is proposed; and the \emph{global flip} dynamics in which at each step a swap of arbitrary spins in the graph is proposed.    We conjecture that both versions mix in polynomial time for the parameters above; we further conjecture that the global flip dynamics mix in time $O(n \log n)$.

\begin{conj}
For $\beta< \beta_c(\Delta)$, the Kawasaki dynamics mix in time polynomial in $n$ for any fixed magnetization and any graph $G$ of maximum degree $\Delta$ on $n$ vertices.

For $\beta > \beta_c(\Delta)$ and $|\eta | > \eta_c(\Delta,\beta)$ the Kawasaki dynamics mix in time polynomial in $n$ for any fixed magnetization $k \ge \eta n$  and any graph $G$ of maximum degree $\Delta$ on $n$ vertices.

For the global flip dynamics, the mixing time in both cases in $O(n \log n)$. 
\end{conj}

In the previous uses of random (bipartite) graphs as gadgets in hardness reductions for approximate counting problems, the gadgets themselves are not in general hard instances for the given problems.   In particular, recent results~\cite{JenssenAlgorithmsJ,helmuth2020finite,chen2021sampling,jenssen2021approximatelySODA} show that for parameters sufficiently deep in the given non-uniqueness regimes, random regular graphs are tractable instances for approximate counting and sampling.   We ask whether for random graphs there are efficient algorithms anywhere inside the NP-hardness regime.
\begin{question}
For $\Delta \ge 3$, $\beta>\beta_c(\Delta)$, is there some $|\eta| < \eta_c(\Delta,\beta)$
so that  there exist efficient approximate counting and sampling algorithms for \csproblem{Fixed-Ising}$(G,\beta,\eta)$ for random $\Delta$-regular graphs?
\end{question}

\subsection{Organization}

In Section~\ref{secPrelim} we provide some of the results we will use in our algorithms and hardness reductions.  In Section~\ref{secHardness} we give the hardness reduction. In Section~\ref{secExtremal} we prove Theorem~\ref{thmExtremal}, solving the extremal problem that identifies the limit of our algorithmic approach.  In Section~\ref{secAlgorithms} we prove the algorithmic results.

\section{Preliminaries}
\label{secPrelim}

Recall that $\cG_\Delta$ denotes the class of graphs of maximum degree $\Delta$.  We use $\mu_{G,\beta,\lam}$ to denote the Ising model on $G$ at inverse temperature $\beta $ and activity $\lam$.  We will often drop $\beta$ from the notation when it remains fixed and we will drop $\lam$ from the notation in the case $\lam=1$ (so $\mu_G = \mu_{G,\beta,1}$ when $\beta$ is understood from the context).   We use the bracket notation $\langle \cdot \rangle_{G,\beta,\lam}$ to denote expectations with respect to the Ising model, in part to distinguish these expectations from expectations over random graphs in Section~\ref{secHardness}.  For a graph $G$, let $\Sigma_G = \{ \pm 1\}^{V(G)}$.  Slightly abusing notation, for $U \subset V(G)$, let $\Sigma_U = \{ \pm 1\}^U$.   We let $M(\sigma)$ denote the magnetization of a configuration $\sigma$ and $X(\sigma)$ denote the number of $+$ spins (so $M(\sigma) = 2 X(\sigma) - |V(G)|$).  We let $\bX$ denote the random variable $X(\sigma)$ when $\sigma$ is drawn from $\mu_{G,\beta,\lam}$.

We now collect a number of results that we will use in the proofs that follow.    The first results are results on zero-free regions for the Ising model partition function, viewed as a (Laurent) polynomial in $\lam$.

\begin{theorem}[Lee--Yang~\cite{lee1952statistical}]\label{thmLeeYang}
  For $\beta \ge 0$, $\lam \in \mathbb C$,  and any graph $G$, $Z_G(\beta,\lam) =0$ only if $|\lam|=1$.
\end{theorem}

\begin{theorem}[Peters--Regts~\cite{PR20}]\label{thmzerofree}
    Let $\Delta\ge 3$ and $\beta\in(0, \beta_c(\Delta))$. Then there exists $\theta=\theta(\beta)\in(0,\pi)$ such that for any $\lam\in\CC$ with $|\arg(\lam)| < \theta$ and any graph $G\in\cG_\Delta$ we have $Z_G(\beta,\lam)\ne 0$.
\end{theorem}

By the following general result of Michelen and Sahasrabudhe, these zero-freeness results imply  central limit theorems for the magnetization of the ferromagnetic Ising model on graphs in $\cG_\Delta$ when $\beta<\beta_c(\Delta)$ or when $\lam >1$.   We apply this result to a random variable counting the number of $+1$ spins in a sample from the Ising model; its generating function is a scaling of $Z_G(\beta,\lam)$.

\begin{theorem}[Michelen--Sahasrabudhe~\cite{MS19}]\label{thm:clt}
    For $n\ge 1$ let $\bX_n$ be a random variable taking values in $\{0,\dotsc,n\}$ with mean $\mu_n$, standard deviation $\sigma_n$, and probability generating function $f_n$. If the roots $\zeta$ of $f_n$ satisfy $|\arg(\zeta)|\ge \delta_n$ and $\sigma_n\delta_n\to\infty$, then $(X_n-\mu_n)/\sigma_n$ converges in distribution to a standard normal random variable.
\end{theorem}

A central limit theorem for the magnetization in fact implies a local central limit theorem, following the approach of Dobrushin and Tirozzi~\cite{dobrushin1977central} (for spin models on $\mathbb Z^d$) and the results of~\cite{jain2021approximate} for the hard-core model.  Let $\bX$ denote the number of $+1$ spins in a sample from the Ising model.  

\begin{prop}
\label{propLocalCLTalg}
Fix $\lam >1$ and $\beta \ge 0$.  Then for any graph $G \in \cG_{\Delta}$ on $n$ vertices and any non-negative integer $\ell$,
\[ \mu_{G, \beta, \lam} \left(  \bX = \ell  \right)    = \frac{1}{\sqrt{2\pi \var(\bX)} } \exp \left[  - \frac{(\ell-  \langle \bX \rangle_{G,\beta, \lam}  )^2  }{2\var(\bX)  }  \right ]   + o(n^{-1/2}) \,,\]
where $\var(\bX) = \langle \bX^2   \rangle_{G,\beta,\lam} -  \langle \bX   \rangle_{G,\beta,\lam}^2$, and where the implied constant in the error term depend only on $\Delta, \beta, \lam$. The same holds for $\beta<\beta_c$, $\lam \ge 1$,  and any $G \in \cG_{\Delta}$.  

Moreover, under the conditions above $\var(\bX) = \Theta(n)$ where again the implied constants depend only on $\Delta, \beta, \lam$.
\end{prop}
We prove Proposition~\ref{propLocalCLTalg} in Appendix~\ref{secAppendixCLT};  the proof of the local central limit theorem is  analogous to that of~\cite[Theorem 1.5]{jain2021approximate} and the proof of the variance bound is  analogous to that of~\cite[Lemma 9]{davies2021approximatelyICALP} and~\cite[Lemma 3.2]{jain2021approximate}. 

For the hardness results, we reduce an \textup{NP}-hard cut problem to the problem of approximating the Ising model at fixed magnetization. 
The $\gamma$-\csproblem{Min-Exact-Balanced-Cut} ($\gamma$-\csproblem{MEBC}) problem is the problem of finding the minimum of  $|E(S, S^c)|$ over all $S\subset V(G)$, $|S| = \lfloor \alpha n \rfloor$, where $n= |V(G)|$.
For stronger inapproximability in Theorem~\ref{thmFixedSupercriticalBeta}\ref{hardFixedSuper}, we apply an inapproximability result due to Bui and Jones~\cite{bui1992finding}, though this is not essential to our method: we can reduce from exactly solving $\gamma$-\csproblem{MEBC} instead.

\begin{theorem}[Bui--Jones~\cite{bui1992finding}]\label{thmMEBC}
Let $\gamma$ be a rational number in $(0,1)$ and let $\eps>0$.  Then $\gamma$-\csproblem{MEBC} is \textup{NP}-hard to approximate within an additive error $n^{2-\eps}$ on $n$-vertex graphs.
\end{theorem}

A key ingredient in the algorithmic results are the efficient approximate counting and sampling algorithms for the ferromagnetic Ising model provided by Jerrum and Sincalir and Randall and Wilson.
\begin{theorem}[Jerrum--Sinclair~\cite{JS93}, Randall-Wilson~\cite{randall1999sampling}]
\label{thmIsingFPRAS}
    For all inverse temperatures $\beta$ and all activities $\lam$, there is an FPRAS and efficient sampling scheme for the Ising model for all graphs $G$.
\end{theorem}

\section{Hardness}
\label{secHardness}

\subsection{The reduction and its properties}\label{secHardnessReduction}

Given $\Delta \ge 3$, $\beta> \beta_c(\Delta)$ and $\eta \in [0, \eta_c)$, 
our goal is to reduce $\gamma$-\csproblem{Min-Exact-Balanced-Cut}  to approximating a fixed-magnetization Ising partition function, for some rational number $\gamma\in((1+\eta/\eta_c)/2, 1)$.  (By symmetry we need only consider $\eta \ge 0$). 

For the reduction we require a gadget $G=G(\Delta,n,\theta,\psi)$ where $\theta,\psi\in(0,1/8)$ are constants that can be determined later in terms of $\Delta,\beta$. 
The gadget is identical to the constructions in~\cite{Sly10,GSVY16}, which is a balanced bipartite graph on $n_G = (2+o(1))n$ vertices. 
The majority of the vertices have degree $\Delta$, and $m=O(n^{\theta})$ vertices on each side of $G$ are designated \emph{terminal vertices} of degree $\Delta-1$.
We detail the construction of the gadget and state its properties after showing how it is used in the reduction.

Let $H$ be a graph on $h = \lfloor n^{\theta/4}/(\Delta-1)\rfloor$ vertices, which is the input for $\gamma$-\csproblem{MEBC}.
Given $G$ as above and an integer $s$, we construct a graph $H^G_s$ of maximum degree $\Delta$ on $ N:= h n_G + s$ vertices as follows:
\begin{itemize}
\item We include a copy $G^x$ of $G$ for each vertex $ x \in V(H)$.
\item We include $s$ isolated vertices.  
\item For each edge $xy \in E(H)$, we include a matching of size $k=\lfloor n^{3\theta/4}\rfloor$ between the left terminals of $G^x$ and left terminals of $G^y$ and a matching of size $k$ between the right terminals of $G^x$ and right terminals of $G^y$.  We do this in such a way that each terminal is used at most once (which is possible since $kh\le m$).
\end{itemize}

For reference, our parameter choices are listed here.  The parameters $\Delta$, $\beta$ are fixed and we can compute $\eta_c$ from them (to arbitrary precision).  The parameter $\eta$ is fixed and satisfies the conditions of the theorem.  From these parameters we compute an arbitrary rational number $\gamma$ such that 
\[ \frac{1+\eta/\eta_c}{2} < \gamma < 1, \]
which is possible because $\eta\in[0,\eta_c)$.
Suitable choices of $\theta,\psi\in(0,1/8)$ can be made in terms of $\Delta,\beta$ (see Lemma~\ref{lemGadgetZ}).
We are then given an instance $H$ of $\gamma$-\csproblem{Min-Exact-Balanced-Cut} on $h$ vertices with $h$ sufficiently large, and we choose an $n$ large enough that $h \le n^{\theta/4}/(\Delta-1)$.
Let  $h_+ = \lfloor \gamma h \rfloor$ and  $h_- = \lceil (1-\gamma) h \rceil$ so that the   $\gamma$-\csproblem{Min-Exact-Balanced-Cut} problem is to find the minimum of $|E(S,S^c)|$ over $S \subset V(H)$ with $|S| = h_+$.  We insist that $h$ is large enough that $\min \{ h_+, h_-\} \ge 1$.  Now let
\begin{itemize}
\item $m=(\Delta-1)^{\lfloor\theta\log_{\Delta-1}n\rfloor}= o(n^{1/8})$,
\item $m'=(\Delta-1)^{\lfloor\theta\log_{\Delta-1}n\rfloor + \lfloor\psi\log_{\Delta-1}n\rfloor} = o(n^{1/4})$,
\item $n_G = 2(n+ m' + m((\Delta-1)^{\lfloor\psi\log_{\Delta-1}n\rfloor}-1)/(\Delta-2)) = (2+o(1))n$
\item $k = \lfloor n^{3\theta/4}\rfloor$,
\item $s$ be a non-negative integer such that
\begin{equation}\label{eqScondition}
\big | 2n(h_+ - h_-) \eta_c   -  \eta [h n_G + s]     \big|  \le \sqrt{nh} \,,
\end{equation}
and $s = \Theta(nh)$,
\item $N = h n_G + s $,
\item $M^* = h_+ - h_-$,
\item $\del>0$ be small enough as a function of $\gamma$ and $\eta_c$,
\item $\ell = \lfloor   N (\eta+1)/2 \rfloor$.
\end{itemize}
The parameters $m,m',n_G,k$ relate to the gadget construction that we detail below. 
There exists an $s$ satisfying~\eqref{eqScondition} with $s=\Theta(nh)$ which  can be found in time polynomial in $h$ because our parameter choices mean that (as $h\to\infty$)
\[ 2n(h_+-h_-)\eta_c = (2+o(1))\cdot (2\gamma-1)\eta_c \cdot nh, \]
and
\[ \eta [hn_G + s] = (2+o(1))\cdot \eta \cdot nh + O(s). \] 
Since $(2\gamma-1)\eta_c > \eta$, for some non-negative integer $s= \Theta(nh)$ the latter can be made within an additive term $O(1)$ of the former (and hence  within $\sqrt{nh}$).  
Finally, $N$ is the number of vertices in the graph $H^G_s$ which we construct in the reduction, $M^*$ is the magnetization of the cuts considered for the $\gamma$-\csproblem{MEBC} problem on $H$, and $\ell$ is such that on $H^G_s$ \csproblem{Fixed-Ising}$(G,\beta,\eta)$ asks for configurations $\sigma$ with $X(\sigma) = \ell$ (and the desired fixed magnetization is thus $2 \ell -N$).

Throughout this section there are many absolute constants (depending only on $\Delta, \beta, \eta$) used and defined.  For ease of reading we will make ample use of $O(\cdot)$ and $\Omega(\cdot)$ notation as well as reusing constants $C, c$ etc.

The main result of this section is the following.

\begin{theorem}\label{thmHardness}
Given $\eps\in(0,1)$ there exists $\zeta > 0$ such that there is a randomized, polynomial-time algorithm to construct a  graph $G$ as above so that with probability at least $2/3$ the following holds: given an $e^{N^\zeta}$-relative approximation to $\Zfixed_{H^G_s} (\beta, 2\ell -N) $ one can compute, in time polynomial in $h$, an additive $h^{2-\eps}$ approximation to the $\gamma$-\csproblem{Min-Exact-Balanced-Cut} of $H$.
\end{theorem}

\noindent
Theorem~\ref{thmHardness} together with Theorem~\ref{thmMEBC} immediately gives Theorem~\ref{thmFixedSupercriticalBeta}\ref{hardFixedSuper}. 

\subsection{The gadget}

We use the same gadget construction as in~\cite{Sly10,GSVY16}.   
The construction is defined by the maximum degree $\Delta$, an integer $n$ and constants $\theta,\psi\in(0,1/8)$ which then determine the parameters $m,m',n_G,k$ listed above.
To construct $G=G(\Delta,n,\theta,\psi)$, let $G'=G'(\Delta,n,\theta,\psi)$ be a random bipartite graph with $n+m'$ vertices on each side obtained by choosing $\Delta$ perfect matchings between the sides uniformly at random, and from the final matching removing $m'$ of the edges. 
With high probability the matchings will be pairwise disjoint sets of edges, so $G'$ is a simple graph.
Let $U_0$ be the set of vertices of degree $\Delta$ in $G'$, and $W_0$ be the set of vertices of degree $\Delta-1$.

To form $G$ from $G'$, on each side partition the $m'$ vertices of degree $\Delta-1$ into $m$ equal-sized sets, and attach the leaves of a copy of a $(\Delta-1)$-ary tree of depth $\lfloor\psi\log_{\Delta-1}n\rfloor$ to each set. 
Then each side of $G'$ has had $m$ trees each of which contains $O(n^\psi)$ vertices added. 
The roots of these trees are now the only vertices of degree $\Delta-1$ in $G$, and there are $m$ roots that were added to each side. 
These are the \emph{terminal} vertices which allow us to connect the gadgets together. 
Let $V_0$ be the vertex set of $G$, and $R_0$ be the terminal vertices. 

Constructed in this way, we want to show that various properties of $G'$ and $G$ hold with sufficiently high probability. 
Many of these properties were verified in~\cite{Sly10,GSVY16}, but we require additional control of statistics of the number of $+1$ spins. 

Throughout this entire section we fix the inverse temperature $\beta>\beta_c$ and take $\lam=1$.  
Recall that $\bX$ denotes the number of $+1$ spins in a sample from the Ising model.      We will  condition on various \emph{phases} of the Ising model on $G$ and on $H^G_s$.   
For $G$, given an Ising configuration $\sigma \in \Sigma_G$, we say the phase is $+$ if $\sum_{v \in U_0} \sigma_v >0$ and $-$ if $\sum_{v \in U_0} \sigma_v <0$.  
If the sum is $0$ then we take the phase to be the spin of some distinguished vertex $u_1 \in U_0$ fixed in advance (arbitrarily).  
Note that neither the spins of $W_0$ nor the spins of the trees added to $G'$ in the construction of $G$ appear in the definition of the phase. 
By symmetry, the probability under $\mu_G$ of each phase is exactly $1/2$.   We denote the Ising model on $G$ conditioned on the $+$ phase and $-$ phase respectively by $\mu_{G,+}$  and $\mu_{G,-}$.   We use $\langle \cdot \rangle _{G,+}$ and $\langle \cdot \rangle _{G,-}$ to denote the corresponding conditional expectation operators.  
For a spin assignment $\tau\in \Sigma_{R_0}$ to the terminals of $G$, we include an additional subscript $\tau$ to denote conditioning on the event $ \{ \sigma_{R_0}=\tau \}$ (that is, $\sigma$ restricted to $R_0$ is equal to $\tau$).

Let $Z_{G,\alpha}$ be the contribution to the partition function $Z_{G}$ of the Ising model on the random gadget $G$ from spin assignments in which there are precisely $2\alpha n$ vertices in $U_0$ of spin $+$.
Let $Z_{G,+}$ and $Z_{G,-}$ be the contributions from the $+$ and $-$ phase respectively to $Z_{G}$. 
We have 
\[ Z_{G,+} = \sum_{1/2 < \alpha \le 1} Z_{G,\alpha} + \frac{1}{2}Z_{G,1/2}, \]
because all contributions with $\alpha>1/2$ belong to the $+$ phase, but for $\alpha=1/2$ we break the tie symmetrically so that half the contribution $Z_{G,1/2}$ goes to the $+$ phase.
Usually, it suffices to use the upper bound on $Z_{G,\pm}$ obtained by taking the entire contribution from $\alpha=1/2$ to the phase at hand.

We now state some results from~\cite{Sly10,GSV15,GSVY16} which are obtained by sophisticated versions of the first and second moment methods and an application of  the small subgraph conditioning method~\cite{robinson1994almost,janson1995random}. 
We will state the results for the $+$ phase, the $-$ phase is completely analogous.
Let $\alpha^+ = (1+\eta_c)/2$, $\alpha^- = (1-\eta_c)/2$, and 
\begin{equation}\label{eqTreeMarginals}
    q =     \frac{(\alpha^+ e^\beta +1-\alpha^+  )^{\Delta-1}   }{(\alpha^+ e^\beta +1-\alpha^+  )^{\Delta-1}  +  (\alpha^+  +(1-\alpha^+)e^\beta  )^{\Delta-1}      },
\end{equation}
which means $q$ is the probability that the root of a $(\Delta-1)$-ary tree gets spin $+1$ in the zero-field `$+$ measure' on the infinite $(\Delta-1)$-ary tree (defined analogously to the $+$ measure on the infinite $\Delta$-regular tree). 
Note that $1/2 < q < \alpha^+$.
We use $Q_S^{+}(\cdot)$ to denote the product measure on $\Sigma_S$ that assigns probability $q$ to $+1$ spins and probability $1-q$ to $-1$ spins, and vice versa for $Q_S^{-}(\cdot)$.

The following lemma collects previous results on the gadget, most notably the near independence of the terminal spins conditioned on a phase.
\begin{lemma}[{\cite[Proof of Theorem 2.1]{Sly10},~\cite[Proof of Lemma~B.3]{GSV15},~\cite[Lemma~22]{GSVY16}}]
    \label{lemGadgetZ}
   Let $G$ be the random graph described above with parameters $\theta,\psi\in(0,1/8)$. Then there exists $c>0$ so that for all $\alpha\in[1/2,1]$,
    \begin{align*}
            \frac{\EE Z_{G,\alpha}}{\EE Z_{G,+}}
            &\le \frac{C}{\sqrt{n}}  e^{-cn(\alpha+\alpha^+)^2},
    \end{align*}
    and for all $\alpha\in[0,1/2]$ 
    \begin{align*}
        \frac{\EE Z_{G,\alpha}}{\EE Z_{G,-}}
        &\le \frac{C}{\sqrt{n}}  e^{-cn(\alpha+\alpha^-)^2}.
\end{align*}
   Moreover, there exist choices of constants $\theta,\psi\in(0,1/8)$ and $C'>0$ so that for large enough $n$, with probability at least $9/10$ over the choice of $G$, all of the following hold simultaneously:
    \begin{enumerate}[(i)]
    \item\label{itmTerminalProductMeasure}
Conditioned on phase $+$, the terminal spins are approximately independent:
\[ \max_{\tau\in \Sigma_{R_0}} \left| \frac{\mu_{G,+} ( \sigma_{R_0} = \tau)}{ Q^+_{R_0}(\tau)} - 1 \right| \le n^{-2\theta}.   \]
  
       \item\label{itmBadTau}  There exists $\mathcal B \subset \Sigma_{W_0}$ so that
       \begin{itemize}
\item $\mu_{G,+}(\sigma_{W_0} \in \mathcal B) \le \exp (-n^{2 \theta}) $
\item For every $\tau_{W_0} \in \Sigma_{W_0}\setminus \mathcal B$, 
\[ \max_{\tau_{R_0} \in \Sigma_{R_0}} \left |  \frac{ \mu_{G,+}  (\sigma_{R_0} = \tau_{R_0}  | \sigma_{W_0} =\tau_{W_0})  } { Q_{R_0}^+(\tau_{R_0})  } -1     \right|  \le n^{-3 \theta} \]
\end{itemize}
        \item\label{itmZlowerBound}
        \[ Z_{G,+} > \frac{1}{C'} \EE Z_{G,+} \,. \]
    \end{enumerate}
    The same also hold with $+$ replaced by $-$. 
\end{lemma}

In previous works~\cite{Sly10,GSVY16} a version of~\ref{itmZlowerBound} with $1/C'$ replaced by a function $o(1)$ as $n\to \infty$ is used (along with the fact that this weaker bound holds with high probability), but the stated version follows from the small subgraph conditioning method used therein~\cite{robinson1994almost,janson1995random}.
In order to handle the fixed-magnetization constraint in our reduction, we show that certain additional properties hold with good probability. 

\begin{lemma}\label{lemGadgetProperties}
For sufficiently large $n$, with probability at least $8/10$ over the choice of the gadget $G$ described above, the following hold simultaneously, for all choices of $\tau \in\Sigma_{R_0}$ (and for  $+$ replaced by $-$ as well):
\begin{enumerate}[(a)]
\item\label{itmFirstMoment}
\[ \langle \left |   \bX    -  2n\alpha^+ \right|  \rangle_{G,+, \tau}  = O( \sqrt{n }  )  \]
\item\label{itmVariance}
\[  \big\langle |\bX - 2n \alpha^+|^2   \big\rangle_{G,+, \tau}   = O(n)    \]
\item\label{itm3rdmoment}
\[  \big\langle |\bX - 2n \alpha^+|^3   \big\rangle_{G,+, \tau}   = O(n^{3/2})    \]
\item\label{itmMGF} 
For  $\del >0$ as specified above and $t_0 = \del/(2 c_0 h)$ for some  constant $c_0 >1/4$,
\begin{align*}
  \big\langle e^{t_0 (\bX - 2\alpha^+n)}   \big\rangle_{G,+, \tau} \le  e^{c_0 t_0^2 n} 
  \intertext{and}
    \big\langle e^{t_0( 2\alpha^+n - \bX)}   \big\rangle_{G,+, \tau} \le  e^{c_0 t_0^2 n}  \, .
\end{align*}
\end{enumerate}
\end{lemma}
We prove Lemma~\ref{lemGadgetProperties} in Section~\ref{lemGadgetProperties}.

\subsection{Proof of Theorem~\ref{thmHardness}}
Here we prove Theorem~\ref{thmHardness} given Lemmas~\ref{lemGadgetZ} and~\ref{lemGadgetProperties}.
Let $G$ be a gadget which satisfies the conclusions of these lemmas, and let $H^G_s$ denote the graph constructed from $H$, $G$ and isolated vertices as above. 
Let $\hat H^G_s$ denote the same graph but without the edges between gadgets (so $\hat H^G_s$ consists of $h$ disjoint copies of $G$ and $s$ isolated vertices). 
We write $\cE$ for the set  of edges $E(H_s^G)\setminus E(\hat H_s^G)$ that lie between gadgets.

Given $\sigma \in \Sigma_{ H^G_s}$, let $Y(\sigma) \in \Sigma_H$ be a vector denoting the phases of the gadgets $G^x$, $x \in V(H)$.  We will call such a $Y$ a \emph{phase vector}.   For a given $Y  \in \Sigma_H$ let $\mu_{H^G_s, Y}$ be the Ising model on $H^G_s$ conditioned on $ \{ Y(\sigma) = Y \}$. Let $\langle  \cdot \rangle_{H^G_s,Y}$ be the corresponding expectation operator. Define $\mu_{\hat H^G_s, Y}$ and $\langle  \cdot \rangle_{\hat H^G_s,Y}$ analogously. 
For $x\in V(H)$ let $R^x$ be the set of terminals in the gadget $G^x$, and let $R$ be the union of the terminal vertices in all the copies of the gadget. 
For a spin assignment $\tau\in \Sigma_R$ to the terminals, we include an additional subscript $\tau$ to indicate conditioning on the event that $\{ \sigma_R=\tau \}$.

We need two probabilistic results before proving Theorem~\ref{thmHardness}.  The first is a large deviation bound  for $\bX$ conditioned on any phase vector $Y$ and any assignment of terminal spins $\tau$.
The second is a local central limit theorem for $\mu_{\hat H^G_s, Y,\tau} (\bX= \ell)$ when $M(Y) = M^* $, and for an arbitrary terminal spin assignment $\tau$.

\begin{lemma}\label{lemLD}
    Assume the gadget $G$ satisfies the conclusions of Lemma~\ref{lemGadgetProperties}.  Then for any phase vector $Y$, any $\tau \in \Sigma_W$,  with
    \[ \nu =   \frac{s}{2} + 2n \sum_{x \in V(H)}  \alpha^{Y_x} \,,   \]
    we have 
    \[  \mu_{\hat H^G_s, Y,\tau} \Big (  \big | \bX - \nu  \big | \ge  \delta n\Big)  \le  \exp \left (- \Omega(n/h) \right )    \,,  \] 
    where $\del>0$ is a constant defined above.
\end{lemma}
\begin{proof}
Note that to leading order $\nu$ is the mean $\langle \bX \rangle_{\hat H^G_s,Y,\tau}$.

We prove the bound on the upper tail; the proof for the lower tail is identical.  Let $\bX_s \sim \mathrm{Bin}(s,1/2)$ be the number of $+$ spins among the $s$ isolated vertices.  Then $\langle e^{t_0 (\bX_s-s/2)} \rangle \le e^{st_0^2/4} \le e^{sc_0 t_0^2}$ since $c_0 >1/4$. 
Using this along with Lemma~\ref{lemGadgetProperties} we can bound the moment generating function,
\begin{align*}
\langle e^{ t_0 (\bX - \nu)} \rangle_{\hat H^G_s, Y,\tau} &\le e^{c_0t_0^2 nh} \,,
\end{align*}
for some constant $c>0$. 
Then we have
\begin{align*}
 \mu_{\hat H^G_s, Y,\tau}  ( \bX \ge \nu + \del n ) &\le    e^{-t_0 \del n} \langle e^{ t_0 (\bX - \nu)} \rangle_{\hat H^G_s, Y,\tau}  \\
 &\le e^{-t_0 \del n + ct_0^2 nh} = e^{- \frac{\del n}{4h}} 
\end{align*}
since $t_0= \del/(2ch)$.
\end{proof}

The next lemma is a local central limit theorem for $\bX$ with respect to $\mu_{\hat H^G_s, Y, \tau}$.
\begin{lemma}\label{lemCLT}
Assume the gadget $G$ satisfies the conclusions of Lemma~\ref{lemGadgetProperties}.  Then for any  phase vector $Y \in \Sigma_H$ with $M(Y) = M^*$, any $\tau \in \Sigma_W$,  and any integer $t$,
\[  \mu_{\hat H^G_s, Y, \tau} \left ( \bX =t \right) =  \frac{1}{\sqrt{2\pi \kappa^2}} \exp \left[  -  \frac{ (t -  \langle \bX \rangle_{\hat H^G_s, Y, \tau} )^2} {2 \kappa^2} \right] +o \left( \frac{1}{\sqrt{nh }} \right) \,,  \]
where $\kappa^2 =   \var_{\hat H^G_s, Y, \tau}(\bX)$. 
In particular, 
\[  \mu_{\hat H^G_s, Y, \tau} \left ( \bX =\ell \right) =  \Omega \left( \frac{1}{\sqrt{nh }} \right) \,.    \]
\end{lemma}
\begin{proof}
We have  $s= \Theta(nh)$, and by the independence of disjoint gadgets and the second moment bound in Lemma~\ref{lemGadgetProperties}, we have $\kappa^2 = \Theta(nh)$.  Moreover,  by our choice of $s$ and the fact that  $M(Y) = M^*$, we have $| \langle \bX \rangle_{\hat H^G_s, Y, \tau} -\ell| = O(\sqrt{nh})$, and so the second statement follows from the first. 

The proof of the first statement is similar to that of Proposition~\ref{propLocalCLTalg} in Appendix~\ref{secAppendixCLT}, but here things are especially simple because of the presence of $s= \Theta(nh)$ isolated vertices.

We start by proving a central limit theorem with the standard method of characteristic functions.  Let $\overline \bX = \frac{1}{\kappa}\big(\bX - \langle \bX \rangle_{\hat H^G_s, Y, \tau} \big)$ and $\phi_{\overline \bX}(t) = \langle e^{it \overline \bX} \rangle_{\hat H^G_s, Y, \tau}$. Then
\begin{align*}
\phi_{\overline \bX}(t) &= \left( \frac{1 + e^{it/\kappa}}{2}  e^{-t/(2 \kappa)} \right)^s  \prod_{x\in V(H)}   \left\langle e^{\frac{it}{\kappa} (\bX - \langle  \bX \rangle_{G, Y_x,\tau_{R^x}})}  \right\rangle_{G, Y_x,\tau_{R^x}} \\
&= \left( 1 - \frac{ t^2 }{8 \kappa^2}  + O(\kappa^{-3}) \right)^s  \prod_{x\in V(H)}  \left( 1- \frac{t^2 \var_{G, Y_x,\tau_{R^x}}(\bX)}{2 \kappa^2}   +O\left(\kappa^{-3} n^{3/2}  \right) \right)  \\
&= e^{-t^2/2} + o(1) \,,
\end{align*}
since $hn^{3/2} \kappa^{-3} = O( h^{-1/2}) \to 0$.  Here we used the bound on the third moment of $\bX$ in a gadget given by Lemma~\ref{lemGadgetProperties}.  This proves that $\overline \bX \Rightarrow N(0,1)$.

 Let $\mathcal L$ denote the lattice $\langle \bX \rangle_{\hat H^G_s, Y, \tau} + \mathbb Z/\kappa$.   Let $\mathcal N(x) = \frac{1}{\sqrt{2 \pi} } e^{-x^2/2}$.   We want to show that 
\[ \sup_{x \in \mathcal L} \left | \kappa \mu_{\hat H^G_s, Y, \tau} ( \overline \bX = x) -  \mathcal N(x) \right| = o(1) \,,\]  
since $\kappa = \Theta(\sqrt{nh})$.  Using Fourier inversion (as in Appendix~\ref{secAppendixCLT}),  we have for any $K>0$,
\begin{align*}
2 \pi \sup_{x \in \mathcal L} &\left |  \kappa\mu_{G,\beta, \lam} ( \overline \bX = x) -  \mathcal N(x) \right|  \\ &= \sup_{x \in \mathcal L}  \left |  \int_{-\pi \kappa }^{\pi \kappa} \phi_{\overline \bX}(t) e^{-itx} \,  dt -  \int_{-\infty}^{\infty} e^{-t^2/2 - itx } \, dt    \right |  \\
&\le  \int_{-\pi \kappa}^{\pi \kappa} \left | \phi_{\overline \bX}(t) - e^{-t^2/2}  \right | \, dt +\int_{|t| > \pi \kappa}e^{-t^2/2  } \, dt  \\
 &\le \int_{-K }^{K} \left | \phi_{\overline \bX}(t) - e^{-t^2/2}  \right | \, dt  + \int_{|t| \ge K}e^{-t^2/2} \, dt + \int_{|t| \ge K}  \left | \phi_{\overline \bX}(t)  \right | \, dt   \\
 &=: A_1 + A_2 +A_3 \,.
\end{align*}

Because $\phi_{\overline X}(t) = e^{-t^2/2} +o(1)$, applying the bounded convergence theorem gives that $A_1 \to 0$ as $n \to \infty$ for any fixed $K$, so we can choose $n$ large enough to guarantee $A_1 < \eps/3$.  We can  pick $K$ large enough to ensure $A_2 < \eps/3$.  For $A_3$, we use the fact that the portion of the characteristic function coming from the isolated vertices has nice behavior.  In particular,
\begin{align*}
| \phi_{\overline X}(t) | &\le \left( \frac{1 + e^{it/\kappa}}{2}  e^{-t/(2 \kappa)} \right)^s \\
&\le e^{-\frac{t^2 s}{4 \kappa^2}} = e^{- \Omega(t^2)} 
\end{align*}
since $ s= \Theta(\kappa^2)$.  Then by choosing $K$ large enough again we can make $A_3 < \eps /3$ as well. \qedhere
\end{proof}

With these ingredients we can prove Theorem~\ref{thmHardness}.

\begin{proof}[Proof of Theorem~\ref{thmHardness}]

    Recall that we assume $\Delta\ge3$, $\beta>\beta_c(\Delta)$, and that $\eta\in [0,\eta_c)$.    Let $G$ be the gadget graph that satisfies the conclusions of Lemmas~\ref{lemGadgetZ} and~\ref{lemGadgetProperties}, and recall the notation of Section~\ref{secHardnessReduction} which includes the graph $H^G_s$ on $N$ vertices formed from copies of $G$ and $s$ isolated vertices.
    
    Let $b$ be the value of $\gamma$-\csproblem{MEBC} on the graph $H$, and let $\eps > 0$.   We will obtain upper and lower bounds on $b$ in terms of $\Zfixed_{H^G_s}(\beta,2\ell - N)$  such that for suitably small $\zeta$, an $e^{N^\zeta}$-relative approximation to $\Zfixed_{H^G_s}(\beta,2\ell - N)$ constrains $b$ to an interval of length at most $h^{2-\eps}$.

    Since $\beta$ and the magnetization $2 \ell - N$ are fixed, we will write $\Zfixed_{H^G_s}$ for $\Zfixed_{H^G_s}(\beta,2 \ell - N)$.  Moreover, for a phase vector $Y \in \Sigma_H$, we write $\Zfixed_{H^G_s}(Y)$ for the contribution to $\Zfixed_{H^G_s}$ from spin assignments  with phase vector $Y$.  For $\tau \in \Sigma_R$ we write $\Zfixed_{H^G_s}(Y,\tau)$ for the contribution to $\Zfixed_{H^G}(Y)$ from spin assignments $\sigma$ which agree with $\tau$ on the terminals $R$.   Similarly, since we only consider the usual Ising model with no external field we write  $Z_{\hat H^G_s}$ for  $Z_{\hat H^G_s}(\beta,1)$.

    We start by bounding the partition function $\Zfixed_{H^G_s}$ from above.  The first step is to split the partition function into sums over $Y$ according to whether $M(Y)=M^*$. We have
    \[ \Zfixed_{H^G_s} = \sum_{Y\in \Sigma_H}\Zfixed_{H^G_s}(Y) = \sum_{Y : M(Y) = M^*}\Zfixed_{H^G_s}(Y) + \sum_{Y : M(Y) \ne M^*}\Zfixed_{H^G_s}(Y). \]
    For an arbitrary phase vector $Y$ we split $\Zfixed_{H^G_s}(Y)$ into a sum over spin assignments $\tau\in \Sigma_R$ to the terminals and pull out the factor of the summand contributed by edges in $\cE$, giving
    \[ \Zfixed_{H^G}(Y) = \sum_{\tau\in \Sigma_R}\Zfixed_{\hat H^G_s}(Y,\tau)\prod_{uv \in \cE}e^{\frac{\beta}{2}\tau_u\tau_v}. \]
    To handle the fixed-magnetization constraint, observe that when $\sigma$ is drawn from the Ising model $\mu_{\hat H^G_s,Y,\tau}$ we have 
    \[ \Zfixed_{\hat H^G_s}(Y,\tau) = Z_{\hat H^G_s}(Y,\tau)\cdot \mu_{\hat H^G_s,Y,\tau}(\bX = \ell), \]
    which we can control with Lemmas~\ref{lemLD} and~\ref{lemCLT}.
    In the case $M(Y)=M^*$ we have 
    \[ \Zfixed_{\hat H^G_s}(Y,\tau) = Z_{\hat H^G_s}(Y,\tau)\cdot \Omega(1/\sqrt{nh}) , \]
    and in the case $M(Y)\ne M^*$ we use 
    \[ \Zfixed_{\hat H^G_s}(Y,\tau) = Z_{\hat H^G_s}(Y,\tau) \cdot \exp(-\Omega(n/h)) . \]
    
    For the sum over $Y$ with $M(Y)=M^*$ this means for some constant $C>0$, 
    \begin{equation}\label{eqUbYbal} 
        \sum_{Y : M(Y)= M^*} \Zfixed_{H^G_s}(Y)
        \le \frac{C}{\sqrt{nh}}\sum_{\tau\in \Sigma_R}Z_{\hat H^G_s}(Y,\tau)\prod_{uv \in \cE}e^{\frac{\beta}{2}\tau_u\tau_v}.
    \end{equation}
    Now we can apply the phase-conditioned, nearly-independent terminal spins property of the gadget.
    Using Lemma~\ref{lemGadgetZ}\ref{itmTerminalProductMeasure} for the inequality (and the fact that $(1+ O(n^{-2\theta}))^{h} = 1+o(1)$), we have
    \[ Z_{\hat H^G_s}(Y,\tau) = Z_{\hat H^G_s}(Y) \cdot \mu_{\hat H^G,Y}(\sigma_R = \tau) \le (1+o(1))Z_{\hat H^G_s}(Y)Q^Y_R(\tau), \]
    where $Q^Y_R(\tau)$ is the probability measure on $\Sigma_R$ such that 
    \[ Q^Y_R(\tau) = \prod_{x\in V(H)} Q^{Y_x}_{R^x}(\tau_{R^x}). \]
    Continuing from~\eqref{eqUbYbal} and absorbing factors into the constant, we have 
    \begin{align*}
        \sum_{Y : M(Y)=M^*} \Zfixed_{H^G_s}(Y) \le \frac{C }{\sqrt{nh}}\sum_{Y : M(Y)=M^*}Z_{\hat H^G_s}(Y) \sum_{\tau\in\Sigma_R} Q^Y_R(\tau)\prod_{uv \in \cE}e^{\frac{\beta}{2}\tau_u\tau_v},
    \end{align*}
    and the final sum over $\tau$ can be expressed in terms of the number $\cut(Y)$ of edges of $H$ which are cut by the phase vector $Y$. 
    This observation appears in~\cite{Sly10} and is precisely why nearly-independent phase-correlated spins are important in reductions such as these. 
    
    Recall $q$ defined in~\eqref{eqTreeMarginals}. For every edge $xy\in E(H)$ cut by $Y$, there are precisely $2k$ edges in $\cE$ such that the measure $Q^Y_W$ gives one endpoint spin $+1$ with probability $q$ and the other endpoint spin $+1$ with probability $1-q$. Such edges are monochromatic with probability $2q(1-q)$.
    Similarly, for every edge $xy\in E(H)$ not cut by $Y$ there are precisely $2k$ edges in $\cE$ which are monochromatic with probability $q^2 + (1-q)^2$. 
    Then for constants $\Theta = 2q(1-q)e^{\beta/2}  + \big(q^2 + (1-q)^2\big)e^{-\beta/2}$ and $\Gamma = 2q(1-q)e^{-\beta/2}  + \big(q^2 + (1-q)^2\big)e^{\beta/2}$ we have
    \[ \sum_{\tau\in \Sigma_R} Q^Y(\tau)\prod_{uv \in \cE}e^{\frac{\beta}{2}\tau_u\tau_v} = \Gamma^{2k|E(H)|}(\Theta/\Gamma)^{2k\cut(Y)}. \]
    Note that   $\Theta < \Gamma$ so that smaller cuts give larger quantities above.
    Finishing the upper bound started in~\eqref{eqUbYbal}, we have 
    \begin{align*} 
        \sum_{Y : M(Y)=M^*} \Zfixed_{H^G_s}(Y) &\le \frac{C}{\sqrt{nh}} \sum_{Y : M(Y)=M^*} Z_{\hat H^G_s}(Y) \Gamma^{2k|E(H)|}(\Theta/\Gamma)^{2k\cut(Y)}
        \\ &\le \frac{C}{\sqrt{nh}}\Gamma^{2k|E(H)|}(\Theta/\Gamma)^{2kb} Z_{\hat H^G_s}, 
    \end{align*}
    because the $\gamma$-\csproblem{MEBC} $b$ of $H$ gives the largest contribution $(\Theta/\Gamma)^{2k\cut(Y)}$, and the partition function $Z_{\hat H^G_s}$ is an upper bound on $\sum_{Y : M(Y)=M^*}Z_{\hat H^G_s}(Y)$.

    For phase vectors $Y$ with $M(Y)\ne M^*$ it suffices to consider the worst-case contribution from edges between gadgets.  For such $Y$, $\left | \frac{s}{2} +  2n \sum_{x \in V(H)}  \alpha^{Y_x}  - \ell   \right|  > \del n$, and so we can apply Lemma~\ref{lemLD} to give
    \begin{align*}
        \sum_{Y : M(Y)\ne M^*} \Zfixed_{H^G_s}(Y)
        &\le 2e^{-\frac{\del n}{4h}}\sum_{\tau\in \Sigma_R}Z_{\hat H^G_s}(Y,\tau)\prod_{uv \in \cE}e^{\frac{\beta}{2}\tau_u\tau_v}
        \\ &\le e^{-\Omega(n/h)} e^{\frac{\beta}{2}|\cE|} \sum_{\tau\in \Sigma_R}Z_{\hat H^G_s}(Y,\tau)
        \\ &\le e^{\beta k|E(H)| -\Omega(n/h)} Z_{\hat H^G_s},
    \end{align*}
    because $|\cE|=2k|E(H)|$.
    Our construction ensures that $k|E(H)| \le kh^2 = o(n/h)$ so that this is a negligible fraction of $Z_{\hat H^G_s}$. 
    Combining these bounds, for all large enough $n$ we have 
    \[ \Zfixed_{H^G_s} \le \left(\frac{C}{\sqrt{nh}} \Gamma^{2k|E(H)|}(\Theta/\Gamma)^{2kb} + e^{\beta k|E(H)| -\Omega(n/h)} \right)Z_{\hat H^G_s}. \]
    The term $\Gamma^{2k|E(H)|}(\Theta/\Gamma)^{2kb}$ is smallest when $b=|E(H)|$, but even in this case it is still at least $e^{-\beta k|E(H)|}$ as $\Theta > e^{-\beta/2}$. 
    Thus, we can absorb the `error' term arising from phase vectors $Y$ with $M(Y)\ne M^*$ into $C$:
    \begin{equation}\label{eqZfixUB}
        \Zfixed_{H^G_s} \le \frac{C}{\sqrt{nh}}\Gamma^{2k|E(H)|}(\Theta/\Gamma)^{2kb} Z_{\hat H^G_s}.
    \end{equation}

    To give a lower bound on $\Zfixed_{H^G}$ it suffices to consider a single phase vector $Y^*$ with $\cut(Y)=b$ that corresponds to the $\gamma$-\csproblem{MEBC} of $H$. Then $M(Y^*)=M^*$ and for some constant $C'>0$ (which will absorb $(1+o(1))$ factors in the calculation below), we have
    \begin{align*}
        \Zfixed_{H^G_s} 
        &\ge \Zfixed_{H^G_s}(Y^*) = \sum_{\tau\in \Sigma_R}\Zfixed_{\hat H^G_s}(Y^*,\tau)\prod_{uv \in \cE}e^{\frac{\beta}{2}\tau_u\tau_v}
        \\ &\ge \frac{C'}{\sqrt{nh}}\sum_{\tau\in \Sigma_R}Z_{\hat H^G_s}(Y^*,\tau)\prod_{uv \in \cE}e^{\frac{\beta}{2}\tau_u\tau_v}
        \\ &\ge \frac{C'}{\sqrt{nh}} Z_{\hat H^G_s}(Y^*) \sum_{\tau\in \Sigma_R}Q^{Y^*}(\tau)\prod_{uv \in \cE}e^{\frac{\beta}{2}\tau_u\tau_v}
        \\ &= \frac{C'}{\sqrt{nh}} Z_{\hat H^G_s}(Y^*) \Gamma^{2k|E(H)|}(\Theta/\Gamma)^{2kb},
    \end{align*}
    where we apply Lemma~\ref{lemCLT} to obtain the second line and the lower bound in~Lemma~\ref{lemGadgetProperties}\ref{itmTerminalProductMeasure} to obtain the third.
    Finally, since we have perfect symmetry between the phases we have $Z_{\hat H^G_s}(Y^*)=2^{-h}Z_{\hat H^G_s}$ and 
    \[ \Zfixed_{H^G_s} \ge \frac{C'2^{-h}}{\sqrt{nh}} \Gamma^{2k|E(H)|}(\Theta/\Gamma)^{2kb} Z_{\hat H^G_s}. \] 

    The upper bound from~\eqref{eqZfixUB} and the lower bound above combine to give 
    \[ \frac{C'2^{-h}}{\sqrt{nh}} \Gamma^{2k|E(H)|}(\Theta/\Gamma)^{2kb} \le \frac{\Zfixed_{H^G_s}}{Z_{\hat H^G_s}} \le \frac{C}{\sqrt{nh}} \Gamma^{2k|E(H)|}(\Theta/\Gamma)^{2kb}, \] 
    which provides the bounds 
    \[ T - \frac{\log(2^h/C')}{2k\log(\Gamma/\Theta)} \le b \le T + \frac{\log C}{2k\log(\Gamma/\Theta)} \] 
    on the min-bisection $b$ of $H$, where 
    \[ T = \frac{\log(Z_{\hat H^G_s}/\Zfixed_{H^G_s}) + 2k|E(H)|\log \Gamma -\log\sqrt{nh}}{2k\log(\Gamma/\Theta)}. \] 

    We can approximate $Z_{\hat H^G_s}$ to within an absolute constant factor in (randomized) time polynomial in $N$ (which is polynomial in $h$) by Theorem~\ref{thmIsingFPRAS}. 
    For the theorem we suppose that we have a relative $e^{N^\zeta}$-approximation of $\Zfixed_{H^G_s}$, and hence if $\tilde T$ is given by the definition of $T$ above with $Z_{\hat H^G_s}$ and $\Zfixed_{H^G_s}$ replaced by these approximate values, we have $|\tilde T - T| \le O(N^\zeta/k)$ and hence
    \[ \tilde T - O\left(\frac{N^\zeta+h}{k}\right) \le b \le \tilde T + O\left(\frac{N^\zeta}{k}\right). \]
    Since $k = \Theta(h^3)$, this constrains $b$ to an interval of length $O\big(N^{\zeta}h^{-3}\big)$. 
    Using $N = O(nh)$ with $n = O(h^{\theta/4})$, it suffices to choose $\zeta$ small enough in terms of $\theta$ and $\eps$ that $(1+4/\theta)\zeta < 5-\eps$.

    Note that we could reduce from solving min-bisection exactly at the cost of weaker inapproximability for $\Zfixed$ in the proof. If $\zeta$ is chosen such that $(1+4/\theta)\zeta < 3$ then the bounds constrain the integer $b$ to an interval of length $o(1)$ and hence for large enough $h$ we can find $b$ exactly.
\end{proof}

\subsection{Proof of Lemma~\ref{lemGadgetProperties}}\label{secGadget}

We prove the lemma in the case of the $+$ phase as the $-$ phase is the same.  Note that we can ignore the contribution of vertices in $R_0$ and the attached trees to $\bX$ in the bounds since  there are $o(n^{1/2})$ of these vertices.  So for this section $\bX$ and $X(\sigma)$ will refer to the number of $+$ spins in the vertices of $G'$.

We first prove the three bounds of the lemma without conditioning on $\{ \sigma_{R_0} =\tau \}$.  We will show that for large enough $n$, with probability at least $8/10$ over the choice of gadget, we have the following bounds:
\begin{align}
\label{eqNoCondbounds1}
 \big \langle |   \bX  -  2n\alpha^+  | \big \rangle_{G,+}  &=O( \sqrt{n } ) \\
 \label{eqNoCondbounds2}
  \big\langle |\bX - 2n\alpha^+|^2   \big\rangle_{G,+}   &= O(n)    \\
   \label{eqNoCondbounds2b}
  \big\langle |\bX - 2n\alpha^+|^3   \big\rangle_{G,+}   &= O(n^{3/2})    \\
  \label{eqNoCondbounds3}
 \big\langle e^{t_0 (\bX - 2\alpha^+n)}   \big\rangle_{G,+} &\le  e^{O( t_0^2 n)} \,.
\end{align}

Let $\xi: \R \to \R$ be a non-negative function that  satisfies $\xi (x) \le e^{c |x|}$ for some constant $c>0$.  We aim to prove bounds on $\langle  \xi (\bX - 2\alpha^+n)   \rangle_{G,+}$ for four choices of functions $\xi$:  $\xi(x) = |x|^k$ for $k \in \{1,2,3\}$, and  $\xi(x) = e^{ t_0 x}$.

For such a function $\xi$, we can write
\begin{align*}
\langle  \xi (\bX - 2\alpha^+n)   \rangle_{G,+}  &\le \frac{  \sum_{\alpha \ge 1/2} \xi (2 \alpha n - 2\alpha^+n) Z_{G,\alpha}    }{ Z_{G,+}   }
\end{align*}
(this is an inequality instead of an equality simply because we include all configurations with $\alpha=1/2$). 

By~\ref{itmZlowerBound} of Lemma~\ref{lemGadgetZ} we have $Z_{G,+} > \frac{1}{C} \EE Z_{G,+}$ with probability at least $1-1/10$.  By Markov's inequality we have 
\begin{align*}
\sum_{\alpha \ge 1/2} \xi (2 \alpha n - 2\alpha^+n) Z_{G,\alpha} &\le 100 \sum_{\alpha \ge 1/2} \xi (2 \alpha n - 2\alpha^+n) \EE Z_{G,\alpha} 
\end{align*}
for all four choices of $\xi$ with probability at least $1-4/100$.  Thus with probability at least $1- 1/10 -4/100 \ge 8/10$ we have 
\begin{align*}
\langle  \xi (\bX - 2\alpha^+n)   \rangle_{G,+}  &\le 100 C  \frac{  \sum_{\alpha \ge 1/2} \xi (2 \alpha n - 2\alpha^+n)  \EE Z_{G,\alpha}    }{ \EE Z_{G,+}   } \,,
\end{align*}
so to prove~\eqref{eqNoCondbounds1}, \eqref{eqNoCondbounds2}, \eqref{eqNoCondbounds2b}, \eqref{eqNoCondbounds3} it is enough to show that $ \frac{  \sum_{\alpha \ge 1/2} \xi (2 \alpha n - 2\alpha^+n)  \EE Z_{G,\alpha}    }{ \EE Z_{G,+}   }$ satisfies the desired bounds.

  Now using the bound
\begin{align*}
\frac{  \EE Z_{G,\alpha}  }{ \EE  Z_{G,+}    }&=   O (n^{-1/2})  e^{-\Omega(n (\alpha  - \alpha^+)^2 )}       \,
\end{align*}
from Lemma~\ref{lemGadgetZ},
we can bound
\begin{align*}
\frac{\sum_{\alpha \ge 1/2}  \xi(2 \alpha n     - 2\alpha^+ n )  \EE Z_{G,\alpha} }{ \EE  Z_{G,+}}  &\le  O(n^{-1/2}) \sum_{n (1-2 \alpha^+)\le  \ell  \le 2n(1-\alpha^+)} \xi(\ell)   e^{-\Omega(\ell^2/n)}  +o(1)\\
&= O(n^{-1/2}) \int_{n (1-2 \alpha^+)} ^{  2n(1-\alpha^+)} \xi(u)  e^{-\Omega(u^2/n)}  \, du  \\
&= O(n^{1/2})  \int_{1-2\alpha^+}^{2(1-\alpha^+)}  \xi( un)  e^{-\Omega(n x^2)}  \, dx \\
&= O \left ( n^{1/2} \int_{-\infty}^{\infty}  \xi( xn)   e^{-\Omega(n x^2)}  \, dx \right)\, .
\end{align*}
Since 
\begin{align*}
\int_{-\infty}^{\infty}  |xn|^k  e^{-\Omega(n x^2)}  \, dx  &= O(n^{(k-1)/2}) 
\intertext{and}
\int_{-\infty}^{\infty}  e^{t_0 nx}  e^{-\Omega(n x^2)}  \, dx  &=e^{O(t_0^2 n)}
\end{align*}
we obtain~\eqref{eqNoCondbounds1}, \eqref{eqNoCondbounds2}, \eqref{eqNoCondbounds2b}, \eqref{eqNoCondbounds3}.

Now we transfer these bounds to the measure conditioned on $ \{ \sigma_{R_0} =\tau \}$. 
For the moment generating function we have
\begin{align*}
 \big\langle e^{t_0 (\bX - 2\alpha^+n)}   \big\rangle_{G,+, \tau} &= \frac{  \big\langle e^{t_0 (\bX - 2\alpha^+n)}  \mathbf 1_{\sigma_{R_0} = \tau}   \big\rangle_{G,+}   }{   \mu_{G,+}  ( \sigma_{R_0} = \tau  )   }  \\
 &\le  \frac{  \big\langle e^{t_0(\bX - 2\alpha^+n)}    \big\rangle_{G,+}   }{    Q^+_{R_0}(\tau) (1 -O(n^{-2 \theta}))  }  \\
 &\le e^{O ( |R_0|  )}  e^{O(t_0^2 n )} \\
 &=  e^{O(t_0^2 n )} \,,
\end{align*}
where we used that $t_0^2n =O(n h^{-2}) = O(n^{1-\theta/2})$ and $|R_0| =O(n^{\theta}) = o(n^{1-\theta/2}) $.

For the $k$th moment,
\begin{align*}
\big\langle |\bX - 2n\alpha^+|^k   \big\rangle_{G,+,\tau} &= \frac{ \big\langle |\bX - 2n\alpha^+|^k  \cdot \mathbf 1_{\sigma_{R_0} = \tau}  \big\rangle_{G,+} }{ \mu_{G,+}  ( \sigma_{R_0} = \tau  )    }  \\
&\le\frac{ \big\langle |\bX - 2n\alpha^+|^k \cdot \mathbf 1_{\sigma_{R_0} = \tau}  \cdot \mathbf 1 _{\sigma_{W_0} \in \mathcal B^c}  \big\rangle_{G,+}   + \big\langle |\bX - 2n\alpha^+|^k   \cdot \mathbf 1 _{\sigma_{W_0} \in \mathcal B}  \big\rangle_{G,+}   }{  Q^+_{R_0}(\tau) (1 -O(n^{-2 \theta})) }  \\
&\le \frac{ \big\langle |\bX - 2n\alpha^+|^k \cdot \mathbf 1_{\sigma_{R_0} = \tau}  \cdot \mathbf 1 _{\sigma_{W_0} \in \mathcal B^c}  \big\rangle_{G,+}   + (2n)^k \exp(-n^{2 \theta})  }{  Q^+_{R_0}(\tau) (1 -O(n^{-2 \theta})) }  \\
&= \frac{ \big\langle |\bX - 2n\alpha^+|^k \cdot \mathbf 1_{\sigma_{R_0} = \tau}  \cdot \mathbf 1 _{\sigma_{W_0} \in \mathcal B^c}  \big\rangle_{G,+}     }{  Q^+_{R_0}(\tau) (1 -O(n^{-2 \theta})) }  + o(1) \\
&= \frac{ \sum_{\sigma \in \Sigma_G}  |X(\sigma) - 2n \alpha^+|^k \mathbf 1_{\sigma_{R_0} = \tau}  \cdot \mathbf 1 _{\sigma_{W_0} \in \mathcal B^c}  \cdot\mu_{G,+}(\sigma)     }{  Q^+_{R_0}(\tau) (1 -O(n^{-2 \theta})) }  + o(1)  \\
&= \frac{  \sum_{\sigma \in \Sigma_G}  |X(\sigma) - 2n \alpha^+|^k \mu(\sigma_{R_0} = \tau| \sigma_{W_0})  \cdot \mathbf 1 _{\sigma_{W_0} \in \mathcal B^c} \cdot\mu_{G,+}(\sigma)     }{  Q^+_{R_0}(\tau) (1 -O(n^{-2 \theta})) }  + o(1)  \\
&=  \frac{ Q^+_{R_0}(\tau) (1 +O(n^{-3 \theta})) \big\langle |\bX - 2n\alpha^+|^k  \cdot \mathbf 1 _{\sigma_{W_0} \in \mathcal B^c}  \big\rangle_{G,+}     }{  Q^+_{R_0}(\tau) (1 -O(n^{-2 \theta})) }  + o(1) \\
&\le (1+ O(n^{-2 \theta}))   \big\langle |\bX - 2n\alpha^+|^k   \big\rangle_{G,+}  + o(1) \\
&= O(n^{(k-1)/2}) \,,
\end{align*}
where we have used from Lemma~\ref{lemGadgetZ} that for all $\tau_{W_0}\in \mathcal B^c$,
\[ \mu_{G,+}  (\sigma_{R_0} = \tau  | \sigma_{W_0} =\tau_{W_0}) = Q_{R_0}^+(\tau) (1+ O( n^{-3 \theta})) \,, \]
and we have used the fact that conditioned on $\sigma_{W_0}$, $X(\sigma)$ and $\sigma_{R_0}$ are independent (using here that in this section $X(\sigma)$ only counts $+$ spins to the vertices of $G'$).

\section{Extremal bounds on the mean magnetization}
\label{secExtremal}

The proof of Theorem~\ref{thmExtremal} is based on that of Krinsky in~\cite{krinsky1975bethe} which applies to lattices such as $\mathbb{Z}^d$.  We require some simple calculus facts recorded in the following lemma.
\begin{lemma}\label{lemCalculus}
Let $\beta>0$, $h_1(x) = \artanh\big(\tanh x\tanh\tfrac{\beta}{2}\big)$, and let $h_2(x,y) = \tanh\big(x + \artanh(\tanh y\tanh\tfrac{\beta}{2})\big)$. 
Then $h_1$ is strictly concave on $(0,\infty)$, $h_2$ is strictly concave when $x,y\in(0,\infty)$, and $h_2(x,x)$ is an increasing function of $x$.
\end{lemma}
\begin{proof}
    For the first statement, note that
    \[ h_1''(x) = -\frac{2 \sinh \beta \sinh (2 x)}{(\cosh \beta+\cosh (2 x))^2}.\]
    Now let $g=\artanh$. 
    For the second statement, note that the Hessian matrix of $h_2$ has determinant
    \[ \frac{\sinh \beta \sinh (2 y)  \tanh \big(x+g(\tanh y \tanh\tfrac{\beta}{2})\big) }{\cosh^2\big(\tfrac{\beta}{2}-y\big) \cosh^2\big(\tfrac{\beta}{2}+y\big)\cosh^4\big(x+g(\tanh y \tanh\tfrac{\beta}{2})\big)} > 0, \] 
    and
    \[ \frac{\partial^2}{\partial x^2}h_2(x,y) = -\frac{2 \sinh \left(g(\tanh y \tanh\tfrac{\beta}{2})+x\right)}{\cosh^3\left(g(\tanh y \tanh\tfrac{\beta}{2})+x\right)} < 0. \]
    This means that the Hessian is negative definite when $x,y\in(0,\infty)$.
    Finally, observe that 
    \[ \frac{\partial}{\partial x}h_2(x,x) = \left(\frac{\sinh \beta}{\cosh \beta+\cosh (2 x)}+1\right) \sech^2\big(x+g(\tanh x\tanh\tfrac{\beta}{2})\big) >0.\qedhere \]
\end{proof}

\begin{proof}[Proof of Theorem~\ref{thmExtremal}]
Let $G=(V,E)$ be a graph and consider the ferromagnetic Ising model with partition function
\begin{align*}
    Z_G(\vec\beta,\lam)  =  \sum_{\sigma \in \Sigma_G}  e^{\frac{1}{2} \sum_{uv \in E(G)} \beta_{uv}\sigma_u \sigma _v} \lam^{M(\sigma)}  \,,
\end{align*}
where we allow each edge $uv\in E$ to have its own inverse temperature parameter $\beta_{uv}$. Specializing to $\beta_{uv} = \beta$ for all edges $uv$, we recover the definition used elsewhere in this work.

When $\beta$ and $\lam$ are understood from context, given a function $f$ with domain $\Sigma_G$, let $\langle f \rangle_G$ be the expected value of $f$ with respect to the Ising model on $G$. 
We also write $\langle f \rangle_{G-uv}$ for the expected value of $f$ with respect to the Ising model on the graph formed from $G$ by removing the edge $uv$, which is equivalent to setting the parameter $\beta_{uv}$ to zero. We extend this notation to $\langle f \rangle_{G-F}$ when we want to remove some set $F$ of edges.

A key feature of the ferromagnetic Ising model with non-negative external field is the following list of Griffiths' inequalities~\cite{griffiths1967correlations}, also known as the GKS inequalities after Kelly and Sherman who generalized Griffiths' work~\cite{kelly1968general}. 
For $A\subset V$, let $\sigma_A = \prod_{v\in A}\sigma_v$. 
Then we have for any graph $G=(V,E)$, $A,B\subset V$, and $uv\in E$,
\begin{align}
    \langle \sigma_A \rangle &\ge 0 \label{eqGr1}\\
    \langle \sigma_A\sigma_B\rangle - \langle \sigma_A\rangle\langle\sigma_B\rangle &\ge 0 \label{eqGr2}\\
    \langle \sigma_A\rangle - \langle \sigma_A\rangle_{G-uv} &\ge 0.\label{eqGr3}
\end{align}
In fact,~\eqref{eqGr2} implies~\eqref{eqGr3} by considering the derivative of $\langle \sigma_A\rangle$ with respect to $\beta_{uv}$. With $B=\{u,v\}$ we have 
\[ \frac{\partial}{\partial \beta_{uv}}\langle\sigma_A\rangle = \frac{1}{2}\langle \sigma_A\sigma_B\rangle - \frac{1}{2}\langle \sigma_A\rangle\langle\sigma_B\rangle \ge 0,\]
where the inequality is by~\eqref{eqGr2}. 

We will apply Griffiths' inequalities during some careful manipulation of expectations using the following identities.
For $s=\pm 1$ and any $\beta$ we have 
\begin{equation}\label{eqExpIdentity}
    e^{\frac{\beta}{2} s} = \left(1+ s\tanh \tfrac{\beta}{2}\right)\cosh \tfrac{\beta}{2},
\end{equation}
which follows from the definitions of the hyperbolic functions in terms of exponential functions. 
We also use the addition formula
\[ \artanh\left(\frac{x+y}{1+xy}\right) = \artanh x + \artanh y. \]

From now on, we work with $\beta_{uv}=\beta$ for all $uv\in E$.
Let $u\in V$ and $v, w\in N(u)$ with $v\ne w$. 
Applying~\eqref{eqExpIdentity} to the term $e^{\tfrac{\beta}{2}\sigma_u\sigma_v}$ which occurs in both the numerator and the denominator of $\langle \sigma_{u}\rangle$ gives
\begin{equation}
    \langle \sigma_u \rangle = \frac{\langle \sigma_u\rangle_{G-uv} + \langle \sigma_{v}\rangle_{G-uv}\tanh \frac{\beta}{2}}{1+\langle\sigma_u\sigma_{v}\rangle_{G-uv}\tanh \frac{\beta}{2}},
\end{equation}
and the same identity applied to the edge $uw$ in $G-uv$ gives
\begin{equation}
    \langle \sigma_u \rangle_{G-uv} = \frac{\langle \sigma_u\rangle_{G-uv-uw} + \langle \sigma_{w}\rangle_{G-uv-uw}\tanh \frac{\beta}{2}}{1+\langle\sigma_u\sigma_w\rangle_{G-uv-uw}\tanh \frac{\beta}{2}}.
\end{equation}
To each of these we apply~\eqref{eqGr2} to the expectation in the denominator, giving
\begin{align}
    \label{eqPreArtanh}
    \langle \sigma_u \rangle &\le \frac{\langle \sigma_u\rangle_{G-uv} + \langle \sigma_{v}\rangle_{G-uv}\tanh\tfrac{\beta}{2}}{1+\langle\sigma_u\rangle_{G-uv}\langle\sigma_{v}\rangle_{G-uv}\tanh\tfrac{\beta}{2}}, \text{ and}
    \\\label{eqPreArtanhuv}
    \langle \sigma_u \rangle_{G-uv} &\le \frac{\langle \sigma_u\rangle_{G-uv-uw} + \langle \sigma_{w}\rangle_{G-uv-uw}\tanh\tfrac{\beta}{2}}{1+\langle\sigma_u\rangle_{G-uv-uw}\langle\sigma_w\rangle_{G-uv-uw}\tanh\tfrac{\beta}{2}}.
\end{align}
Applying $g=\artanh$ to both sides of~\eqref{eqPreArtanh} and using the addition formula, we obtain
\begin{equation}\label{eqMagu}
    g(\langle \sigma_u \rangle) \le g(\langle \sigma_u\rangle_{G-uv}) + g(\langle \sigma_v\rangle_{G-uv}\tanh\tfrac{\beta}{2}).
\end{equation}
Doing this again with~\eqref{eqPreArtanhuv}, we also use~\eqref{eqGr3} with the edge $uw$ in the graph $G-uv$, giving $\langle \sigma_w\rangle_{G-uv-uw} \le \langle \sigma_w\rangle_{G-uw}$ and hence
\begin{equation}\label{eqToIterate}
    g(\langle \sigma_u \rangle_{G-uv}) \le g(\langle \sigma_u\rangle_{G-uv-uw}) + g(\langle \sigma_w\rangle_{G-uw}\tanh\tfrac{\beta}{2}).
\end{equation}

Observe that we can iterate the process used to obtain~\eqref{eqToIterate} over each $w\in N(u)\setminus\{v\}$ in turn to obtain
\begin{equation}\label{eqPreSymmetry}
    g(\langle \sigma_u \rangle_{G-uv}) \le \log\lam + \sum_{w\in N(u)\setminus \{v\}}g(\langle\sigma_{w}\rangle_{G-uw}\tanh\tfrac{\beta}{2}),
\end{equation}
where we have used the fact that removing the set $F_u := \{uv : v\in N(u)\}$ of edges incident to $u$ we have
\[ \langle \sigma_u\rangle_{G-F_u} = \frac{\lam -\lam^{-1}}{\lam +\lam^{-1}} = \tanh (\log \lam) \]because $u$ is an isolated vertex in $G-F_u$.

At this point, Krinsky assumes that $G$ is both edge and vertex transitive so that for some $L>0$ and for all $uv\in E$ we have 
\[ \langle \sigma_u \rangle_{G-uv} = \langle \sigma_{v} \rangle_{G-uv} = \tanh L.\]
In this special case,~\eqref{eqPreSymmetry} becomes
\[ L \le \log\lam + (\Delta-1)g(\tanh L\tanh\tfrac{\beta}{2}), \]
where $\Delta$ is the degree of a vertex in $G$. This implies that $L$ is bounded above by the largest solution $L^*$ to 
\begin{equation}\label{eqSymmetryL*}
    L^* = \log\lam + (\Delta-1)g(\tanh L^*\tanh\tfrac{\beta}{2}).
\end{equation}
Plugging this into~\eqref{eqMagu} and observing that for any vertex-transitive graph $\langle\sigma_u\rangle$ is the mean magnetization $\eta_G$, we have 
\[ \eta_G \le \tanh\big(L^*+g(\tanh L^*\tanh\tfrac{\beta}{2})\big). \]
The right-hand side is precisely $\eta^+_{\Delta,\beta,\lam}$, the mean magnetization of the $+$ measure on the infinite $\Delta$-regular tree, which one can derive from first principles (as in, e.g.,~\cite{Bax82}). 
In fact, it suffices to observe that every inequality we applied to obtain this bound holds with equality in the tree. 
That is, in the infinite $\Delta$-regular tree $\langle\sigma_u\sigma_{v}\rangle_{G-uv} = \langle\sigma_u\rangle_{G-uv}\langle\sigma_{v}\rangle_{G-uv}$ since removing $uv$ leaves $u$ and $v$ in different connected components so $\sigma_u$ and $\sigma_v$ are independent. 
Similarly, we have $\langle\sigma_u\sigma_w\rangle_{G-uv-uw} = \langle\sigma_u\rangle_{G-uv-uw}\langle\sigma_w\rangle_{G-uv-uw}$ and since $w$ is in a different component from the edge $uv$ after $uw$ is removed, $\langle \sigma_w\rangle_{G-uv-uw} = \langle \sigma_w\rangle_{G-uw}$.
The tree is edge and vertex transitive, proving that the derived upper bound is given by the mean magnetization of some measure on the tree. 
As the $+$ measure stochastically dominates all other translation-invariant measures on the tree, it corresponds to the largest solution $L^*$ to~\eqref{eqSymmetryL*}.

We now will apply this argument to a finite graph $G$ that is not necessarily vertex or edge transitive. 
First, we observe that we can reduce to the case that $G$ is regular by a well-known construction.
Suppose that $G$ has maximum degree $\Delta$ but minimum degree $\delta\le \Delta-1$. 
We construct a graph $H$ with minimum degree $\delta+1$ such that $\eta_G \le \eta_H$. 
Let $H_0$ be formed from the disjoint union of two copies of $G$, so that the mean magnetization of $H_0$ is equal to that of $G$, 
For $i\ge 0$, if there is a vertex $u$ of degree $\delta$ in $H_i$ let $H_{i+1}$ be formed from $H_i$ by connecting $u$ to its copy in the other copy of $G$.
The inequality~\eqref{eqGr3} shows that $\eta_{H_i}$ is non-decreasing as $i$ increases because adding the edge can only increase any term $\langle\sigma_v\rangle$, and the mean magnetization is the average of these terms over all vertices $v$.
When the process terminates at some $H$, the minimum degree of $H$ is $\delta+1$, and we cannot have decreased the mean magnetization. 
Iterating this construction, we can obtain a $\Delta$-regular graph $H$ whose mean magnetization is an upper bound on the mean magnetization of $G$, hence it suffices to prove the theorem in the case of a $\Delta$-regular graph.

In a $\Delta$-regular graph, careful averaging and applications of Jensen's inequality allow us to recover the result obtained for edge and vertex transitive graphs.
We can interpret~\eqref{eqPreSymmetry} as a property of an oriented edge $\overrightarrow{uv}$, and average over all edges incident to $u$ oriented away from $u$. Let $L_{\overrightarrow{uv}}$ be given by $\tanh L_{\overrightarrow{uv}}=\langle \sigma_u \rangle_{G-uv}$, so that averaging~\eqref{eqPreSymmetry} over $v\in N(u)$ gives
\begin{equation}\label{avgNu}
    \frac{1}{\Delta}\sum_{v\in N(u)} L_{\overrightarrow{uv}} \le \log\lam + \frac{\Delta-1}{\Delta}\sum_{v\in N(u)} g(\tanh L_{\overrightarrow{vu}}\tanh\tfrac{\beta}{2}).
\end{equation}
By Lemma~\ref{lemCalculus}, the function $x \mapsto g\big(\tanh x\tanh\tfrac{\beta}{2}\big)$ is concave on $(0,\infty)$.
This means that \eqref{avgNu} and Jensen's inequality give 
\begin{equation}\label{afterJensen}
    \frac{1}{\Delta}\sum_{v\in N(u)} L_{\overrightarrow{uv}} \le \log\lam +(\Delta-1) g\bigg(\tanh\Big[\frac{1}{\Delta}\sum_{v\in N(u)}L_{\overrightarrow{vu}}\Big]\tanh\tfrac{\beta}{2}\bigg).
\end{equation}
To clean this up, we define
\begin{align*} 
    A_u &:= \frac{1}{\Delta}\sum_{v\in N(u)} L_{\overrightarrow{uv}},&
    B_u &:= \frac{1}{\Delta}\sum_{v\in N(u)} L_{\overrightarrow{vu}}, 
\end{align*}
so that~\eqref{afterJensen} gives 
\begin{equation}\label{AuBu}
    A_u \le \log\lam +(\Delta-1) g\left(\tanh B_u \tanh\tfrac{\beta}{2}\right).
\end{equation}
This we average over a uniform random $u\in V$ and again appeal to concavity. 
Here we finally obtain the desired equation because the averages satisfy 
\[ \overline L := \frac{1}{\Delta n}\sum_{uv\in E}\left(L_{\overrightarrow{uv}} + L_{\overrightarrow{vu}}\right) = \frac{1}{\Delta n}\sum_{u\in V}\sum_{v\in N(u)}L_{\overrightarrow{uv}} = \frac{1}{\Delta n}\sum_{u\in V}\sum_{v\in N(u)}L_{\overrightarrow{vu}}, \]
so an application of Jenssen's inequality gives for $\overline L$ what we had for $L$ in the case of a transitive graph,
\[ \overline L \le \log\lam + (\Delta-1)g(\tanh \overline L\tanh\tfrac{\beta}{2} ). \] 
As before, this means that $\overline L\le L^*$.

To conclude the argument in the $\Delta$-regular case, we apply the same averaging trick to~\eqref{eqMagu}. 
For any edge $uv\in E$,
\[
    g(\langle \sigma_u \rangle) \le g(\langle \sigma_u\rangle_{G-uv}) + g(\langle \sigma_{v}\rangle_{G-uv}\tanh\tfrac{\beta}{2}),
\] 
so fixing $u$ and averaging over $v\in N(u)$ gives
\[
    g(\langle \sigma_u \rangle) \le \frac{1}{\Delta}\sum_{v\in N(u)}g(\langle \sigma_u\rangle_{G-uv}) + \frac{1}{\Delta}\sum_{v\in N(u)}g(\langle \sigma_{v}\rangle_{G-uv}\tanh\tfrac{\beta}{2}).
\] 
Applying Jensen's inequality again, we have
\[
    g(\langle \sigma_u \rangle) \le A_u + g(\tanh B_u\tanh\tfrac{\beta}{2}),
\] 
and so
\begin{equation}\label{eqsigmau}
    \langle \sigma_u \rangle \le \tanh(A_u + g(\tanh B_u\tanh\tfrac{\beta}{2})).
\end{equation}
By Lemma~\ref{lemCalculus}, the right-hand side is concave as a function of $A_u$ and $B_u$.
Averaging~\eqref{eqsigmau} over $u\in V$ and applying Jensen's inequality, we conclude 
\begin{align*}
    \eta_G = \frac{1}{n}\sum_{u\in V}\langle \sigma_u \rangle 
    &\le \tanh( \overline L + g(\tanh \overline L\tanh\tfrac{\beta}{2}))\\
    &\le \tanh( L^* + g(\tanh L^*\tanh\tfrac{\beta}{2})) = \eta_{\Delta,\beta,\lam}^+,
\end{align*}
using the fact that $x\mapsto \tanh\big(x + g(\tanh x\tanh\tfrac{\beta}{2})\big)$ is non-decreasing proved in Lemma~\ref{lemCalculus}.
\end{proof}

\section{Algorithms}
\label{secAlgorithms}

In this section we prove Theorems~\ref{thmAlgFixedSubcriticalBeta} and~\ref{thmFixedSupercriticalBeta}\ref{algFixedSuper}. By symmetry, it suffices to consider the case when $\eta\ge0$.  We can exclude the trivial case $\eta =1$ since there is just a single spin configuration in that case.

We will use several ingredients from Section~\ref{secPrelim}.  Fix $\beta, \eta, \Delta$ satisfying the conditions of either theorem, and let $G$ be a graph of maximum degree $\Delta$ on $n$ vertices.  Let $\ell = \lfloor \frac{1+\eta}{2}  n\rfloor$ so that our goal is to sample an Ising configuration $\sigma$ with $X(\sigma) = \ell$.  

Since we can efficiently sample from $\mu_{G,\beta,\lam}$ for any $\lam$ via Theorem~\ref{thmIsingFPRAS}, we can perform a binary search on values of $\lam$, estimating $\langle  \bX \rangle_{G,\beta,\lam}$, to find a $\lam$ so that
\begin{equation}
\label{eqLamclose}
 \left|  \langle  \bX \rangle_{G,\beta,\lam}   - \ell  \right|  = o(\sqrt{n}) \,. 
 \end{equation}
Given such an activity $\lam$, we will  approximately sample from $\mu_{G,\beta,\lam}$ until we sample a configuration $\sigma$ with  $X(\sigma) =\ell$ and then output $\sigma$.  For this algorithm to be efficient we must ensure that the probability of hitting this value is not too small; in fact, we will show that it is $\Theta(n^{-1/2})$.    

For Theorem~\ref{thmAlgFixedSubcriticalBeta} this follows immediately from Proposition~\ref{propLocalCLTalg} which provides a local central limit theorem and $\Theta(n)$ variance for $\bX$ for $\beta< \beta_c(\Delta)$ and any activity $\lam$.   

For Theorem~\ref{thmFixedSupercriticalBeta}, we need to ensure that we can find $\lam$  satisfying~\eqref{eqLamclose} that is bounded away from $1$ independently of $n$ so we can apply Proposition~\ref{propLocalCLTalg}.   This is guaranteed by the extremal result, Theorem~\ref{thmExtremal}, and the conditions of the Theorem~\ref{thmFixedSupercriticalBeta}.  In particular, because $\eta > \eta_+(\beta, \Delta)$ (and by continuity of the magnetization of the $+$ measure on the tree) there is some $\lam_{\min} > 1$ so that $\eta = \eta^+_{\Delta, \beta, \lam_{\min}}$.  Theorem~\ref{thmExtremal} then says that to achieve mean magnetization $\eta$ on any $G \in \cG_{\Delta}$ we must take $\lam \ge \lam_{\min}$, thus giving the required uniform bound away from $1$.

In what follows we give the details of the approach.  We note that the running time of our algorithm could certainly be improved by using a faster Ising sampler (e.g.~\cite{mossel2013exact}) or by using the techniques of~\cite{jain2021approximate}, but here we will not try to optimize the running time beyond finding polynomial-time algorithms.

The existence of the efficient sampling schemes in  Theorems~\ref{thmAlgFixedSubcriticalBeta} and~\ref{thmFixedSupercriticalBeta}\ref{algFixedSuper} are proved in Theorem~\ref{thmSampleK} in Section~\ref{subsecSample}. The existence of approximate counting algorithms follows from the sampling algorithms via a standard reduction. We provide the details in Appendix~\ref{secCounting}.

\subsection{Bounds on the activity}
Here we prove a lemma guaranteeing the existence of a good activity $\lam$ for use in our sampling algorithms.
We write $\lam^{-1}_{\Delta,\beta}(\eta)$ for the value of $\lam$ such that $\eta^+_{\beta,\Delta,\lam}=\eta$.  In particular, when $\eta > \eta_+ (\beta,\Delta)$ we have $\lam^{-1}_{\Delta,\beta}(\eta) > 1$.

\begin{lemma}\label{lemExistLambda}
    Let $\Delta\ge 3$, $\beta< \beta_c(\Delta)$, and $\eta \in [0,1)$ be fixed. Let $\lam_{\min} = 1$ and $\lam_{\max}=\sqrt{\frac{1+\eta}{1-\eta}} e^{\beta\Delta/2}$.
    Then for any $G \in \cG_\Delta$ on $n$ vertices, there exists an integer $t \in \{ 0, \dots, \lfloor (\lam_{\max} - \lam_{\min})n \rfloor \}$ so that for $\ell = \lfloor \frac{\eta+1}{2} n  \rfloor $ and $\lam_t =\lam_{\min} +  \frac{t}{n}$ 
    we have
    \begin{equation}
        \label{eqCloseLam}
        \left |  \langle \bX \rangle_{G, \beta, \lam_t} - \ell  \right | = O(1)\,,
    \end{equation}
    where the implied constant depends only on $\Delta, \beta, \eta$. 
   The same holds for $\beta> \beta_c$ and $\eta \in (\eta_+ (\beta,\Delta), 1)$ with $\lam_{\min} = \lam^{-1}_{\Delta,\beta}(\eta) > 1$ and $\lam_{\max}=\sqrt{\frac{1+\eta}{1-\eta}} e^{\beta\Delta/2}$.
\end{lemma}

Before we prove Lemma~\ref{lemExistLambda}, we need one simple bound on the magnetization of a bounded-degree graph.

\begin{lemma}\label{lemLamBound}
    For all $\Delta \ge 1$, $\beta \ge 0$, $\eta \in [0,1)$, and any $G \in \cG_{\Delta}$, the  value of $\lam$ such that the  mean magnetization $\eta_G(\beta,\lam)$ is exactly $\eta$ satisfies 
    \[1  \le \lam \le  \sqrt{\frac{1+\eta}{1-\eta}} e^{\beta\Delta/2} \,. \]
\end{lemma}
\begin{proof}
  The lower bound follows from symmetry of the Ising model.  For the upper bound, let $\sigma\in \Sigma_G$ be drawn from $\mu_{G,\beta,\lam}$ and suppose that $\eta=\eta_G(\beta,\lam)$. 
    Let $u$ be a uniform random vertex of $G$, and let $Y=M(\sigma_{N(u)})$ be the magnetization of the neighbors of $u$.
    That is, when $u$ has exactly $j$ neighbors with spin $+$, $Y=2j-\deg(u)$.  
    Then a direct computation in the Ising model gives
    \[ \eta = \left\langle\frac{\lam e^{\frac{\beta}{2}Y} - \lam^{-1}e^{-\frac{\beta}{2}Y}}{\lam e^{\frac{\beta}{2}Y} + \lam^{-1}e^{-\frac{\beta}{2}Y}}\right\rangle_{G,\beta,\lam}= \left\langle\frac{\lam^2 - e^{-\beta Y}}{\lam^2 + e^{-\beta Y}}\right\rangle_{G,\beta,\lam}. \] 
    Since $Y\ge-\Delta$ and the function of $Y$ inside the expectation is increasing for $\beta>0$, we immediately obtain
    \[ \eta \ge \frac{\lam^2-e^{\beta\Delta}}{\lam^2 + e^{\beta\Delta}} \Longleftrightarrow \lam \le \sqrt{\frac{1+\eta}{1-\eta}} e^{\beta\Delta/2}.  \qedhere\]
\end{proof}

We now prove  Lemma~\ref{lemExistLambda}.
\begin{proof}[Proof of Lemma~\ref{lemExistLambda}]

By  Theorem~\ref{thmExtremal} and  Lemma~\ref{lemLamBound}, the value of $\lam$ such that $\eta_G(\beta,\lam)=\eta$ satisfies $ \lam_{\min} \le\lam\le \lam_{\max}$. 
A standard calculation gives that
\[ \frac{\partial}{\partial\lam}\eta_G(\beta,\lam) = \frac{2}{n\lam}\var(\bX) \,. \] 
Then by Proposition~\ref{propLocalCLTalg}, we have $\frac{\partial}{\partial\lam}\eta_G(\beta,\lam)   >0 $ and $\frac{\partial}{\partial\lam}\eta_G(\beta,\lam) = O(1)$.
This means that $\eta_G(\beta,\lam)$ and hence $\langle \bX \rangle_{G,\beta,\lam}$ are strictly increasing as functions of $\lam$, and that $\langle \bX \rangle_{G,\beta,\lam}$ can increase by at most $O(1)$ on any interval of length $1/n$. 
\end{proof}

\subsection{The sampling algorithm}
\label{subsecSample}

Our algorithm is the combination of a simple binary search on values of $\lam$ and repeated sampling from distributions $\hat\mu_{\beta,\lam}$ that approximate the usual Ising model $\mu_{G,\beta,\lam}$. 

\medskip
\textbf{Algorithm: Sample-$k$}
\begin{itemize}
    \item INPUT: $\Delta,\beta,\eta,\eps$ and $G \in \cG_\Delta$ on $n$ vertices.
    \item OUTPUT: $\sigma\in\Sigma_G(k)$ distributed within $\eps$ total variation distance of the fixed-magnetization Ising model $\nu_{G,\beta,k}$, where $k$ is the largest integer such that $k\equiv n\mod 2$ and $k\le\eta n$.
\end{itemize}

\begin{enumerate}
    \item Let $\lam_{\min}, \lam_{\max}$ be as given in Lemma~\ref{lemExistLambda}.
    \item For $t = 0, \dots, \lfloor (\lam_{\max}-\lam_{\min}) n \rfloor$, let $\lam_t = \lam_{\min} + t/n$.     \item Let $\Lambda_0 = \{\lam_t : t = 0, \dots,   \lfloor (\lam_{\max}-\lam_{\min}) n \rfloor\}$. 
    \item FOR $i=1, \dots , C\log n$,
          \begin{enumerate}
              \item Let $\lam$ be a median of the set $\Lambda_{i-1}$.
              \item With $N=C'n^2\log\big(\frac{\log n}{\eps}\big)$, take $N$ independent samples $\sigma_1,\dotsc$, $\sigma_N$ from a distribution $\hat \mu_{G,\beta,\lam}$ on $\Sigma_G$ such that $\| \hat \mu_{G,\beta,\lam}  -   \mu_{G,\beta,\lam}  \|_{TV} < \eps' =  \frac{1}{C N \log n}$.
               \item If there exists $j \in \{1, \dots , N\}$ so that $M(\sigma_j)=k$, then output $\sigma_j$ for the smallest such $j$ and HALT.
              \item Let $\overline k  = \frac{1}{N}  \sum_{j=1}^N M(\sigma_j)$.
             
              \item If $\overline k \le k$, let $\Lambda_i = \{ \lam' \in \Lambda_{i-1} : \lam' > \lam \}$.  If instead $\overline k > k$, let $\Lambda_i = \{ \lam' \in \Lambda_{i-1} : \lam' < \lam \}$.
          \end{enumerate}

    \item If no spin assignment of magnetization $k$ has been obtained by the end of the FOR loop (or if $\Lambda_j = \emptyset$ at any step), output a spin assignment of magnetization $k$ by taking the first $(k+n)/2$ vertices in an arbitrary ordering on $V(G)$ and setting their spins to $+$ and remaining spins to $-$.
\end{enumerate}

\begin{theorem}\label{thmSampleK}
    Let $\hat \nu_{G,\beta,k}$ be the output distribution of the algorithm Sample-$k$.  Then for $n$ large enough, 
    \begin{enumerate}
\item $\|   \hat \nu_{G,\beta,k} - \nu_{G,\beta,k} \|_{TV} <\eps$.
\item The running time of Sample-$k$ is polynomial in $n$ and $\log(1/\eps)$.
\end{enumerate}
\end{theorem}
This proves the approximate sampling portions of Theorems~\ref{thmAlgFixedSubcriticalBeta} and~\ref{thmFixedSupercriticalBeta}\ref{algFixedSuper}. 

\begin{proof}[Proof of Theorem~\ref{thmSampleK}]
We first give the proof under the assumption that each $\hat\mu_{\beta,\lam}$ is precisely the Ising model $\mu_{G,\beta,\lam}$, and subsequently we will reduce the general case that $\hat\mu_{\beta,\lam}$ is close to $\mu_{G,\beta,\lam}$ to this case with a standard coupling argument. 

We say a \emph{failure} occurs at step $i$ in the FOR loop if either 
\begin{enumerate}[(1)]
    \item\label{itmFail1} $|n\eta_G(\beta,\lam)-\overline k| > 1/4$, or
    \item\label{itmFail2} $|n\eta_G(\beta,\lam)-k| \le 1/2$ but none of the samples $\sigma_1,\dotsc,\sigma_N$ have magnetization exactly $k$ and so the algorithm did not HALT on line (d).
\end{enumerate}
Note that avoiding~\ref{itmFail1} means that the sample mean $\overline k$ of the $N$ spin assignments sampled is close to the true mean $n\eta_G(\beta,\lam)$, and avoiding~\ref{itmFail2} means that in the case that the true mean magnetization is close to $k$, we successfully sample a spin assignment of the desired magnetization $k$.
To establish that the probability of a failure occurring at any step is at most $\eps/2$, it suffices to establish that the probability of each type of failure in a given step is at most $\eps/(4C\log n)$.
Consider an arbitrary step $i$ with the value of $\lam$ assigned for that step in line (a), and note that $\overline k$ is the mean of $N$ independent samples from $\hat\mu_{\beta,\lam}$. Under the assumption that $\hat\mu_{\beta,\lam}=\mu_{G,\beta,\lam}$, we have $\EE\overline k=n\eta_G(\beta,\lam)$, so by Hoeffding's inequality we have 
\[ \Pr(|n\eta_G(\beta,\lam)-\overline k| > 1/4) \le 2e^{-N/(32n^2)}. \]
This is at most the desired $\eps/(4C\log n)$ when $N\ge \Omega(n^2\log(\log(n)/\eps))$.

For the second type of failure, we suppose that the current value of $\lam$ means that $|n\eta_G(\beta,\lam)-k| \le O(1)$, but that none of the $N$ samples from $\hat\mu_{\beta,\lam}$ give a state with magnetization exactly $k$. 
Each `trial' to get a state of magnetization $k$ is independent at succeeds with probability $p\ge \Omega(1/\sqrt{n})$ by Lemma~\ref{propLocalCLTalg}. Then we have no successful trials with probability $(1-p)^N$, which is at most $\eps/(4C\log n)$ when $N\ge \Omega(\sqrt{n}\log(\log(n)/\eps))$.

The above lower bounds on $N$ show that the value given in line (b) suffices. With a bound of $\eps/2$ on the probability of any failure, we now show that the output state has distribution within total variation distance $\eps/2$ of $\nu_{G,\beta,k}$ on $\Sigma_G(k)$. 
If no failure occurs then the algorithm must reach a value of $\lam$ such that $|n\eta_G(\beta,\lam)-k| \le 1/2$. This is because the starting search set $\Lambda_0$ contains such a value by Lemma~\ref{lemExistLambda}, and by binary search structure of the algorithm.
Note that if there is no failure then $\overline k$ is an accurate representation of $n\eta_G(\beta,\lam)$, so we continue searching in the larger half of the search set when $\overline k < k-1/4$ and so $n\eta_G(\beta,\lam) \le \overline k + 1/4 < k$. The case that $\overline k>k+1/4$ is similar. 
The desired total variation distance now follows from the fact that for any $\lam'$, $\nu_{G,\beta,k}$ is precisely $\mu_{G,\beta,\lam'}$ conditioned on getting a spin assignment of magnetization exactly $k$. 
Under the assumption that $\hat\mu_{\beta,\lam}=\mu_{G,\beta,\lam}$, this means that if the algorithm outputs a state on line (d) during the FOR loop, then the output distribution is precisely $\nu_{G,\beta,k}$. 
Since we have proved that a state is output on line (d) of some step with probability at least $\eps/2$, this is equivalent to showing that when $\hat\mu_{\beta,\lam}=\mu_{G,\beta,\lam}$ the output distribution is within total variation distance $\eps/2$ of $\nu_{G,\beta,k}$.

We do not need to assume  access to efficient algorithms for sampling from $\mu_{G,\beta,\lam}$ exactly, we can make do with good approximate samplers given by Theorem~\ref{thmIsingFPRAS}. 
A standard interpretation of total variation distance is that when each $\hat\mu_{\beta,\lam}$ has total variation distance $\xi$ from $\mu_{G,\beta,\lam}$, there is a coupling between the measures such that the probability they disagree is at most $\xi$. 
Then to prove Theorem~\ref{thmSampleK} in general we can add a third failure condition: that any sample from any $\hat\mu_{\beta,\lam}$ disagrees with $\mu_{G,\beta,\lam}$ under this coupling. We make at most $C N\log n$ calls to any approximate sampling algorithm, so by a union bound this type of failure occurs with probability at most $\eps/2$ when we have the stated total variation distance $\xi \le \eps/(2CN\log n)$. 
Now, with the above proof in the case of exact samplers we output a state distributed according to $\nu_{G,\beta,k}$ with probability at least $1-\eps/2$, and with approximate samplers as described we get the same output unless the third type of failure occurs. 
That is, with probability at least $1-\eps$ we output a state distributed according to $\nu_{G,\beta,k}$. 
In terms of total variation distance, this is an $\eps$-approximate sampler for $\nu_{G,\beta,k}$ with running time $O(N\log n\cdot T(n,\eps))$ as desired.
\end{proof}

\newcommand{\etalchar}[1]{$^{#1}$}

\appendix

\section{Proof of Proposition~\ref{propLocalCLTalg}}
\label{secAppendixCLT}

First we prove the variance bound.

\begin{proof}
Fix $\Delta, \beta, \lam$ satisfying the assumptions of the proposition, and let $G \in \cG_{\Delta}$ be a graph on $n$ vertices. 
    For the lower bound, fix an independent set $J$ in $G$ of size at least $n/(\Delta+1)$, let $U=V(G)\setminus J$, and let $\sigma_U, \sigma_J$ be $\sigma$ restricted to $U$ and $J$ respectively.
    Then $\var(\bX) \ge \langle \var(\bX|\sigma_U) \rangle_{G,\beta,\lam}$ by the law of total variance.
    But conditioned on $\sigma_U$, the spins of the vertices in $J$ are independent, and moreover, the conditional marginal probability of each spin taking $+$ is bounded away from $0$ and $1$ independently of $n$.
    That is, conditioned on $\sigma_U$ the indicator random variable $\bX_v = \mathbf 1_{\sigma_v =1}$ for $v$ in $J$ are mutually independent Bernoulli random variables with
    \begin{align*}
        \var(\bX_v|\tau) & = \frac{1}{( \lam e^{\frac{\beta}{2}Y_{v,\tau}} + \lam^{-1}e^{\frac{-\beta}{2}Y_{v,\tau}} )^2},
    \end{align*}
    where $Y_{v,\tau}$ is the magnetization under $\sigma_U$ of the neighbors of $v$ (so $Y_{v,\tau} = 2j-\deg(v)$ when $v$ has $j$ neighbors with spin $+$ under $\sigma_U$).
    Since we have $\beta \ge0$ and $\lam\ge1$, it is straightforward to prove that $\var(\bX_v|\tau)$ is minimized when $Y_{v,\tau}=\deg(v)$, and is a decreasing function of $\deg(v)$.
    We conclude that
    \[ \var(\bX|\sigma_U) = \sum_{v\in J}\var(\bX_v|\sigma_U) \ge |J|\frac{1}{(\lam e^{\frac{\beta \Delta}{2}} + \lam^{-1} e^{-\frac{\beta \Delta}{2}})^2}. \]
    
    For the upper bound, we proceed as in~\cite[Lemma 3.2]{jain2021approximate}.  Let $\hat Z_G(\beta,t) = \lam^n Z_G(\beta,\lam)$  with $t = \sqrt{\lam}$.  That is,
    \begin{align*}
\hat Z_G(\beta,t) &= \sum_{\sigma \in \Sigma_G} e^{\frac{\beta}{2} \sum_{(u,v) \in E(G)} \sigma_u \sigma _v} \cdot t^{ X(\sigma)} \,.
\end{align*}
This defines the same Ising model, but will be slightly easier to work with since with $\beta$ fixed,  the partition function  $\hat Z_G(\beta, t)$ is a polynomial in $t$ of degree $n$.  Let $\xi_1, \dots, \xi_n$ denote its complex roots.  The assumptions of Proposition~\ref{propLocalCLTalg}, along with Theorems~\ref{thmLeeYang} and~\ref{thmzerofree}, imply that there exists some $\del>0$, depending only on $\Delta$, $\beta$, and $\lam$ so that for each $j$, 
\begin{equation}
\label{eqDelRoots}
| \xi_j - t| > \del
\end{equation}
 with $| \cdot |$ denoting distance in the complex plane.  We write 
   \begin{align*}
\hat Z_G(\beta,t) &= \prod_{j=1}^n (1- t/\xi_j) 
\end{align*}
and then
    \begin{align*}
\var(\bX) &= t^2   \frac{\partial^2 \log \hat Z_G(\beta,t}{ \partial t^2}   + t  \frac{\partial \log \hat Z_G(\beta,t}{ \partial t}  \\
&= - t \sum_{j=1}^n \frac{1} { \xi_j (1- t /\xi_j)^2} \,.
\end{align*}
From Theorem~\ref{thmLeeYang} we have $|\xi_j| =1$ for all $j$ and from~\eqref{eqDelRoots} we have $|1- t /\xi_j|^{-2} \le \del^{-2}$, giving $\var(\bX) \le  \frac{n \sqrt{\lam}}{\del^2}$. 
    \end{proof}

Next we prove the local central limit theorem.  As remarked above, this essentially follows from the method of~\cite{dobrushin1977central}, as it is straightforward to generalize their proof from $\mathbb Z^d$ to general bounded-degree graphs.
\begin{proof}
Again fix $\Delta, \beta, \lam$ satisfying the assumptions of the proposition, and let $G \in \cG_{\Delta}$ be a graph on $n$ vertices.   We let $\langle \cdot \rangle$ denote expectation with respect to $\mu_{G,\beta,\lam}$.

 Let $\phi_{\bX}(t) = \langle  e^{it \bX} \rangle$ denote the characteristic function of $\bX$, and let   $\kappa^2 = \langle (\bX - \langle \bX \rangle)^2 \rangle$ be the variance of $\bX$.    Let  $\overline \bX = ( \bX - \langle \bX \rangle)/\kappa$, and let $\phi_{\overline \bX}(t)$ be the characteristic function of $\bX$.   Let $\mathcal L$ denote the lattice $\langle \bX \rangle + \mathbb Z/\kappa$.   Let $\mathcal N(x) = \frac{1}{\sqrt{2 \pi} } e^{-x^2/2}$.   We want to show that 
\[ \sup_{x \in \mathcal L} \left |  \mu_{G,\beta, \lam} ( \overline \bX = x) -  \frac{1}{\kappa}\mathcal N(x) \right| = o(n^{-1/2}) \,,\]  
or 
\[ \sup_{x \in \mathcal L} \left | \kappa \mu_{G,\beta, \lam} ( \overline \bX = x) -  \mathcal N(x) \right| = o(1) \,,\]  
since $\kappa = \Theta(\sqrt{n})$.  Using Fourier inversion  we have
\begin{align*}
2 \pi \sup_{x \in \mathcal L} &\left |  \kappa\mu_{G,\beta, \lam} ( \overline \bX = x) -  \mathcal N(x) \right|  \\ &= \sup_{x \in \mathcal L}  \left |  \int_{-\pi \kappa }^{\pi \kappa} \phi_{\overline \bX}(t) e^{-itx} \,  dt -  \int_{-\infty}^{\infty} e^{-t^2/2 - itx } \, dt    \right |  \\
&\le  \int_{-\pi \kappa}^{\pi \kappa} \left | \phi_{\overline \bX}(t) - e^{-t^2/2}  \right | \, dt +\int_{|t| > \pi \kappa}e^{-t^2/2  } \, dt  \\
 &\le \int_{-K }^{K} \left | \phi_{\overline \bX}(t) - e^{-t^2/2}  \right | \, dt  + \int_{|t| \ge K}e^{-t^2/2} \, dt + \int_{|t| \ge K}  \left | \phi_{\overline \bX}(t)  \right | \, dt   \\
 &=: A_1 + A_2 +A_3 \,.
\end{align*}
It is enough to show that for every $\eps>0$, there is $n$ large enough so that $A_1 + A_2 +A_3 \le\eps$.  To do this we will choose $K$ large enough as a function of $\eps$. 

Because $\lim_{n \to \infty} \phi_{\overline X}(t) = e^{-t^2/2}$  for every fixed $t$ (by the central limit theorem of Theorem~\ref{thm:clt} and the zero-freeness results supplied by the conditions of the proposition and Theorems~\ref{thmLeeYang} and~\ref{thmzerofree}), applying the bounded convergence theorem gives that $A_1 \to 0$ as $n \to \infty$, so we can choose $n$ large enough to guarantee $A_1 < \eps/3$. 

We can pick $K$ large enough to ensure $A_2 < \eps/3$.   Similarly, if we can show that there exists some $c>0$ so that for every $t \in \mathbb R$, $|\phi_{\overline \bX}(t) | \le e^{-c t^2}$, then we can choose $K$ large enough to ensure $A_3 < \eps/3$, as desired.  

To show this, it is enough to prove that $|\langle e^{it \bX}  \rangle|  \le e^{-c n t^2}$.  To do this, we proceed as in the proof of the variance lower bound as above.    Let $J \subset V(G)$ be an independent set  of size at least $n/(\Delta+1)$  and let $U = V(G) \setminus J$.  We will again use the fact that conditioned on $\sigma_U$, the spins of the vertices in $J$ are independent with bounded marginals.   Let $\bX = \bX_{U} + \bX_J$ where $\bX_U, \bX_J$ are the number of $+$ spins in the vertices in $U$ and $J$ respectively. Then
\begin{align*}
|\langle e^{it \bX} \rangle| & = |\langle e^{it \bX_{U}} e^{it \bX_J} \rangle| \\
&\le \langle  | e^{it \bX_{U}} \langle  e^{it \bX_J} | \sigma_{U} \rangle | \rangle \\
&\le \langle  |\langle   e^{it \bX_J} | \sigma_{U} \rangle |\rangle  \\
&\le \max_{\tau_U}  |\langle   e^{it \bX_J} | \sigma_{U} = \tau_U \rangle | \,.
\end{align*}
We will show that for all spin assignments $\tau_U$ to $U$, $ |\langle   e^{it \bX_J} | \sigma_{U} = \tau_U \rangle | \le e^{-cn t^2}$.   This follows from the independence and bounded marginals properties of the spins in $J$ conditioned on $\sigma_U$.  In particular, there exists $c'>0$ so that $| \langle e^{it (1+\sigma_v)/2 } | \sigma_U=\tau_U \rangle| \le e^{-c't^2}$ over all $v \in J$ and all choices of $\tau_{U}$.  Then by independence, $|\langle e^{it \bX_J} | \sigma_{U=\tau_U} \rangle | \le e^{-c'n t^2 / (\Delta+1)}$, and the claim follows by taking $c = c'/(\Delta+1)$.

\end{proof}

\section{Approximate counting via sampling}
\label{secCounting}

In this section we fix $\Delta\ge 3$, $\beta\in(\beta_c(\Delta),1)$, $\eta\in[0,1)$, an $n$-vertex graph $G\in\cG_\Delta$, and $k$ such that $0\le k\le \eta n$ and $n\equiv k\mod 2$. Since $G$ and $k$ remains fixed we freely drop them from notation such as $\Zfixed_G(\beta,k)$ and $\nu_{G,\beta,k}$. 
Reducing approximate counting to approximate sampling is a well-studied area, and a standard simulated annealing approach lets us do this for $\Zfixed_G(\beta)$ and $\nu_{\beta}$. 
For simplicity, we present only a basic form of the argument. More sophisticated cooling schedules exist that could be used to improve the running time of the reduction. 
The key observations behind this application of simulated annealing are that $\Zfixed(1)=\binom{n}{(n+k)/2}$, and for any valid magnetization $k$ and any parameters $\beta$, $\beta'$ we have 
\begin{equation}\label{eqAnnealRatio}
    \frac{\Zfixed(\beta')}{\Zfixed(\beta)} = \left\langle e^{\frac{\beta'-\beta}{2}\delta(\sigma)}\right\rangle_{\beta,k},
\end{equation}
where $\delta(\sigma) = \sum_{uv\in E(G)}\sigma_u\sigma_v$ and $\langle\cdot\rangle_{\beta,k}$ represents an expectation over $\nu_{G,\beta,k}$.
These facts follow straightforwardly from the definition of $\Zfixed(\beta)$. 
Given a \emph{cooling schedule} 
\[ 1 = \beta_0 < \beta_1 < \dotsb < \beta_\ell = \beta, \]
and the independent random variables $R_i = e^{\frac{\beta_{i+1}-\beta_i}{2}\delta(\sigma_i)}$ where $\sigma_i\sim\nu_{\beta_i}$, we have 
\[ \Zfixed(\beta) = \binom{n}{(n+k)/2} \prod_{i=0}^{\ell-1}\langle R_i\rangle_{\beta,k}. \]
Moreover, by~\eqref{eqAnnealRatio} we can estimate $\langle R_i\rangle_{\beta,k}$ by taking $S$ repeated independent samples from a $\xi$-approximate sampler $\hat\nu_i$ for $\nu_{\beta_i}$, and setting $\hat R_i$ to be the mean of the $S$ samples. 
For the cooling schedule to be useful in an FPRAS, we need to determine suitable $S$ and $\xi$ so that
\begin{equation}\label{eqRAS}
    \Pr\left((1-\eps)\Zfixed_G(\beta,k) \le\binom{n}{(n+k)/2}\prod_{i=0}^\ell \hat R_i \le (1+\eps)\Zfixed_G(\beta,k)\right) \ge 3/4.
\end{equation} 

An important consideration for a cooling schedule is the quantity
\[ \frac{\langle R_i^2\rangle_{\beta,k}}{\langle R_i\rangle_{\beta,k}^2} = \frac{\Zfixed(2\beta_{i+1}-\beta_i)\Zfixed(\beta_i)}{\Zfixed(\beta_{i+1})^2}, \]
and when this is bounded above by a constant $B$ for all $0\le i\le \ell$ we say the schedule is \emph{$B$-Chebyshev}.
A particularly simple cooling schedule has $\beta_i = i\log(1+1/n)$ for $i<\ell$, so that for some length $\ell = \Theta(n)$ we have a schedule where $\beta_{i+1}-\beta_i \le \log(1+1/n)$.
Then from the fact that when $r\ge 1$ we have $\Zfixed(\beta + r) \le r^{\Delta n/2}\Zfixed(\beta)$, it is straightforward to show that $1\le R_i\le e^{\Delta/2}$ and 
\[ \frac{\langle R_i^2\rangle_{\beta,k}}{\langle R_i\rangle_{\beta,k}^2} \le (1+1/n)^{\Delta n/2} \le e^{\Delta/2}, \]
meaning that the simple cooling schedule is $e^{\Delta/2}$-Chebyshev. 
A second-moment argument, given for example in~\cite[Section~2]{SVV09} and in~\cite{JS97} (but originating in~\cite{DF91}), gives that some $S=\Omega(\ell/\eps^2)$ and $\xi=O(\eps/\ell)$ suffice for~\eqref{eqRAS} (where the implied constants can depend on $\Delta$).
By the algorithm from Section~\ref{secAlgorithms}, we can obtain such $\xi$-approximate samples in time polynomial in $n$ and $\log(1/\eps)$. The total running time for the approximate counting algorithm is therefore polynomial in $n$ and $1/\eps$, meaning we have the desired FPRAS\@.

This simple cooling schedule works in great generality, but there are a few well-known ways of improving the efficiency of the method. 
The first is to use a more complex (in fact, \emph{adaptive}) but shorter schedule, as given in~\cite{SVV09}.
A second improvement makes use of \emph{warm starts}, which is essentially a strategy of reusing randomness when a Markov chain is the basis of the approximate samplers $\hat\nu_i$ which are used to obtain $\hat R_i$. 
These improvements are combined in~\cite[Section 7]{SVV09}.

\end{document}